%% file: main.tex
\documentclass[a4paper]{elsarticle}

\usepackage{amsfonts,amsthm,amsmath,amssymb}
\usepackage{mathtools}
\usepackage{tensor}
\usepackage{nicefrac}
\usepackage{multicol,multirow}
\usepackage{makecell}
\usepackage[textwidth=3cm]{todonotes}
\usepackage{subcaption}
\usepackage{upgreek}
\usepackage{hyperref}
\usepackage{microtype}

\usepackage{tikz}
\usetikzlibrary{patterns,arrows,backgrounds,calc,shapes,
shadows,decorations.pathmorphing,decorations.pathreplacing,automata,
shapes.multipart,positioning,shapes.geometric,fit,circuits,trees,
shapes.gates.logic.US,fit,decorations.markings}

\newtheorem{example}{Example}
\newtheorem{constraints}{Constraints}
\newtheorem{definition}{Definition}
\theoremstyle{remark}

\newtheoremstyle{mythmstyle}%
    {}%
    {}%
    {\it}%
    {}%
    {\bf}%
    {}%
    { }%
    {\thmname{#1}\thmnumber{ #2\addcontentsline{toc}{subsubsection}{{\it#1 #2}}}%
    \thmnote{ (#3)}. }
\theoremstyle{mythmstyle}
\newtheorem{lemma}{Lemma}
\newtheorem{theorem}{Theorem}
\newtheorem{proposition}{Proposition}

\journal{Journal of Computer and System Sciences}

\input{macros}

\begin{document}

\begin{frontmatter}

\title{The Complexity of Reachability in Parametric~Markov~Decision~Processes\tnoteref{funding}}
\tnotetext[funding]{This work was supported by the DFG RTG 2236
``UnRAVeL'', and the ERC Advanced Grant 787914 ``FRAPPANT'', and by NSF grants 1545126 (VeHICaL) and 1646208, by the DARPA Assured Autonomy program, by Berkeley Deep Drive, and by Toyota under the iCyPhy center.  Most work was done while the first author was with RWTH Aachen University.}

\author[sebiaddress]{Sebastian Junges\corref{mycorrespondingauthor}}
\address[sebiaddress]{University of California at Berkeley, USA}
\cortext[mycorrespondingauthor]{Corresponding author}
\ead{sjunges@berkeley.edu}

\author[rwth]{Joost-Pieter Katoen}
\address[rwth]{RWTH Aachen University, Germany}
\ead{katoen@cs.rwth-aachen.de}

\author[uantwerp]{Guillermo A. P\'erez}
\address[uantwerp]{University of Antwerp, Belgium}
\ead{guillermoalberto.perez@uantwerpen.be}

\author[rwth]{Tobias~Winkler}
\ead{tobias.winkler@cs.rwth-aachen.de}

\begin{abstract}
  This article presents the complexity of reachability decision problems for
  parametric Markov decision processes (pMDPs), an extension to Markov
  decision processes (MDPs) where transitions probabilities are described by
  polynomials over a finite set of parameters.  In particular, we study the
  complexity of finding values for these parameters such that the induced MDP
  satisfies some maximal or minimal reachability probability constraints. We
  discuss different variants depending on the comparison operator in the
  constraints and the domain of the parameter values. We improve all known
  lower bounds for this problem, and notably provide ETR-completeness results
  for distinct variants of this problem.
\end{abstract}

\begin{keyword}
  Parametric Markov decision processes \sep Formal verification \sep
  Existential theory of the reals \sep Computational complexity \sep Parameter synthesis
\end{keyword}

\end{frontmatter}

%\linenumbers

\section{Introduction}
\input{1_introduction}

\section{Preliminaries}
\input{2_preliminaries}

\section{Parametric Markov decision processes}
\input{3_pmdp}

\section{Qualitative Reachability Problems}
\input{4_qualitative}

\section{Quantitative Reachability Problems}
\input{5_quantitative}

\section{Conclusions}
\input{6_conclusion}

\section*{Acknowledgements}
We thank Els Hoekstra for useful feedback on a previous version of this
article.

\bibliography{literature}

% Don't forget to remove toc for final version
%\clearpage
%\tableofcontents

\end{document}

%% file: macros.tex
\newcommand{\ifemptyelse}[3]{\ifthenelse{\equal{#1}{}}{#2}{#3}}

\newcommand{\Distr}{\ensuremath{\mathnormal{Distr}}}

\newcommand{\supp}{\ensuremath{\mathnormal{supp}}}
\newcommand{\QQ}{\ensuremath{\mathbb{Q}}}
\newcommand{\QQnn}{\ensuremath{\QQ_{\geq 0}}}
\newcommand{\QQpos}{\ensuremath{\QQ_{> 0}}}
\newcommand{\QQneg}{\ensuremath{\QQ_{< 0}}}
\newcommand{\RR}{\ensuremath{\mathbb{R}}}

\newcommand{\NN}{\ensuremath{\mathbb{N}}}

\newcommand{\quant}{\ensuremath{\mathcal{Q}}}
\newcommand{\true}{\ensuremath{\mathtt{true}}}
\newcommand{\false}{\ensuremath{\mathtt{false}}}

\newcommand{\vars}{\ensuremath{\boldsymbol{x}}}
\newcommand{\nrvars}{\ensuremath{k}}

\newcommand{\val}{\ensuremath{\mathsf{val}}}
\newcommand{\proj}[2]{\ensuremath{#1}}
\newcommand{\Val}[1][]{\ensuremath{\mathsf{Val}_{#1}}}
\newcommand{\wdval}[1]{\ensuremath{\mathsf{Val}^{\mathrm{wd}}_{#1}}}
\newcommand{\gpval}[1]{\ensuremath{\mathsf{Val}^{\mathrm{gp}}_{#1}}}
\newcommand{\pureval}[1]{\mathbb{B}}

\newcommand{\vanish}[2][]{\ensuremath{\mathsf{Vanish}_{#1}(#2)}}

\newcommand{\graph}[1]{\ensuremath{G(#1)}}
\newcommand{\edges}[1]{\ensuremath{E(#1)}}

\newcommand{\first}[1]{\ensuremath{\mathsf{first}(#1)}}
\newcommand{\last}[1]{\ensuremath{\mathsf{last}(#1)}}
\newcommand{\state}[2]{\ensuremath{#2}@{#1}}
\newcommand{\pathsfin}[1]{\ensuremath{\Pi^{#1}_\mathsf{fin}}}
\newcommand{\pathsinf}[1]{\ensuremath{\Pi^{#1}_{\infty}}}
\newcommand{\pathset}[1]{\ensuremath{\Pi^{#1}}}

\newcommand{\eventually}[2][]{\ensuremath {\lozenge^{\ifemptyelse{#1}{}{\leq #1}}{#2}}}

\newcommand{\mdp}{\ensuremath{M}}
\newcommand{\dtmc}{\ensuremath{D}}
\newcommand{\pmdp}{\ensuremath{\mathcal{M}}}
\newcommand{\pmc}{\ensuremath{\mathcal{D}}}

\newcommand{\Act}{\ensuremath{\mathnormal{Act}}}
\newcommand{\act}{\ensuremath{\alpha}}
\newcommand{\altact}{\ensuremath{\beta}}
\newcommand{\init}{\ensuremath{\iota}}

\newcommand{\pars}{\ensuremath{X}}
\newcommand{\pmdptuple}{\ensuremath{(S,\init,\Act,\pars,\pprob)}}

\newcommand{\regenc}[2]{\ensuremath{\Upphi^{#1}_{#2}}}

\newcommand{\prob}{\ensuremath{P}}
\newcommand{\pprob}{\ensuremath{\mathcal{P}}}

\newcommand{\EnAct}{\ensuremath{\mathbb{A}\mathrm{ct}}}

\newcommand{\sched}{\ensuremath{\sigma}}

\newcommand{\Sched}[1]{\ensuremath{\mathfrak{S}^#1}}
\newcommand{\MSched}[1]{\ensuremath{\mathfrak{S}_m^{#1}}}
\newcommand{\DMSched}[1][]{\ensuremath{\Sigma^#1}}
\newcommand{\induced}[2]{\ensuremath{#1[#2]}}
\newcommand{\instantiated}[2]{\ensuremath{#1[#2]}}

\newcommand{\generator}[2][]{\ensuremath{\langle #2 \ifemptyelse{#1}{}{\mid #1} \rangle}}
\newcommand{\region}{\ensuremath{R}}

\newcommand{\etral}[1]{\ensuremath{\Upphi^{\unrhd,\exists}_\mathrm{wd}(#1)}}
\newcommand{\etrau}[1]{\ensuremath{\Upphi^{\unlhd,\exists}_\mathrm{wd}(#1)}}
\newcommand{\etrdl}[1]{\ensuremath{\Upphi^{\unrhd,\forall}_\mathrm{wd}(#1)}}
\newcommand{\etrdu}[1]{\ensuremath{\Upphi^{\unlhd,\forall}_\mathrm{wd}(#1)}}
\newcommand{\etra}[1]{\ensuremath{\Upphi^{\bowtie,\exists}_\mathrm{wd}(#1)}}
\newcommand{\etrd}[1]{\ensuremath{\Upphi^{\bowtie,\forall}_\mathrm{wd}(#1)}}

\newcommand{\target}{\ensuremath{T}}

\newcommand{\prsym}{\ensuremath{\mathnormal{Pr}}}
\newcommand{\prpath}[2]{\ensuremath{\prsym^{#1}(#2)}}
\newcommand{\pr}[4][]{\mathnormal{Pr}_{#2}^{#3}(\ifemptyelse{#1}{}{#1 \models }#4)}

\newcommand{\prsol}[3][]{\tensor*[]{\mathsf{sol}}{^{#3}_{\ifemptyelse{#1}{}{#1 \rightarrow }#2}}}
\newcommand{\prsolgp}[3][]{\tensor*[_{\mathrm{gp}}^{\mathbb{P}}]{\mathsf{sol}}{^{#3}_{\ifemptyelse{#1}{}{#1 \rightarrow }#2}}}

\newcommand{\pzerostates}[2][]{\ensuremath{S^{#2\ifemptyelse{#1}{}{,#1}}_{=0}}}
\newcommand{\ponestates}[2][]{\ensuremath{S^{#2\ifemptyelse{#1}{}{,#1}}_{=1}}}
\newcommand{\pzeroexstates}[2][]{\ensuremath{S^{#2\ifemptyelse{#1}{}{,#1}}_{\exists=0}}}
\newcommand{\poneexstates}[2][]{\ensuremath{S^{#2\ifemptyelse{#1}{}{,#1}}_{\exists=1}}}
\newcommand{\pzeroallstates}[2][]{\ensuremath{S^{#2\ifemptyelse{#1}{}{,#1}}_{\forall=0}}}
\newcommand{\poneallstates}[2][]{\ensuremath{S^{#2\ifemptyelse{#1}{}{,#1}}_{\forall=1}}}

\newcommand{\cclass}[1]{\text{\textbf{#1}}}
\newcommand{\Ptime}{\cclass{P}}
\newcommand{\NP}{\cclass{NP}}

\newcommand{\PSPACE}{\cclass{PSPACE}}
\newcommand{\ETR}{\cclass{ETR}}

\newcommand{\problem}[1]{\text{\textsc{#1}}}
\newcommand{\reachdp}[1]{\ensuremath{\problem{reach}^{#1}}}

\newcommand{\ceq}{=}%\ensuremath{\doteq}}
\newcommand{\cleq}{\leq}%\ensuremath{\leqslant}}
\newcommand{\cgeq}{\geq}%\ensuremath{\geqslant}}
\newcommand{\cgt}{>}%\ensuremath{\gtrdot}}
\newcommand{\clt}{>}%\ensuremath{\lessdot}}
\newcommand{\enc}{\ensuremath{\mathsf{enc}}}

\newcommand{\control}{\ensuremath{c}}

\tikzset{
  dot hidden/.style={},
  line hidden/.style={},
  dice hidden/.style={},
  dot color/.style={dot hidden/.append style={color=#1}},
  dot color/.default=black,
  line color/.style={line hidden/.append style={color=#1}},
  line color/.default=black,
  dice color/.style={dice hidden/.append style={color=#1,fill}},
  dice color/.default=white
}\def\dotsize{0.1}
\newcommand{\drawdie}[2][]{%
\begin{tikzpicture}[x=0.9em,y=0.9em,#1]
  \draw 	[thick, rounded corners=0.5,line hidden,dice hidden] (0,0) rectangle (1,1);
  \ifodd#2
    \fill[dot hidden] (0.5,0.5) circle (\dotsize);
  \fi
  \ifnum#2>1
  \fill[dot hidden] (0.25,0.25) circle (\dotsize);
  \fill[dot hidden] (0.75,0.75) circle (\dotsize);
  \ifnum#2>3
    \fill[dot hidden] (0.25,0.75) circle (\dotsize);
    \fill[dot hidden] (0.75,0.25) circle (\dotsize);
    \ifnum#2>5
      \fill[dot hidden] (0.75,0.5) circle (\dotsize);
      \fill[dot hidden] (0.25,0.5) circle (\dotsize);
    \fi
  \fi
\fi
\end{tikzpicture}%
}

\newcommand{\refsize}[1]{{\scriptsize #1}}

%% file: 1_introduction.tex
Markov decision processes (MDPs) are \emph{the} model to reason about sequential processes under (stochastic) uncertainty and non-determinism.
Markov chains (MCs) are MDPs without non-determinism.
Often, probability distributions in these models are difficult to assess precisely during design time of a system.
This shortcoming has led to interval
MCs~\cite{DBLP:conf/lics/JonssonL91,DBLP:journals/ipl/ChenHK13,DBLP:conf/rp/Sproston18,DBLP:conf/tacas/SenVA06}
and interval MDPs (also known as bounded-parameter MDPs)~\cite{givan2000bounded,DBLP:journals/ai/WuK08,DBLP:conf/cav/PuggelliLSS13}, which allow for interval-labelled transitions.
Analysis under interval Markov models is often too pessimistic:
The actual probabilities on the transitions are considered to be non-deterministically and \emph{locally} chosen.
Intuitively, consider the probability of a coin-flip yielding heads in some uncertain environment. In interval models, the probability may vary with the local memory state of an agent acting in this environment. Such behaviour is unrealistic.
\emph{Parametric} MCs and MDPs~\cite{DBLP:conf/ictac/Daws04,DBLP:journals/fac/LanotteMT07,DBLP:conf/nfm/HahnHZ11,DBLP:journals/ai/DelgadoSB11} (pMCs, pMDPs) overcome this limitation by adding dependencies (or couplings) between various transitions---they add global restrictions to the selection of the probability distributions.
Intuitively, the probability of flipping heads can be arbitrary, but should be independent of an agent's local memory. 
Such couplings are similar to restrictions on schedulers in decentralised/partially observable MDPs, considered in e.g.,~\cite{DBLP:journals/mor/BernsteinGIZ02,DBLP:journals/tcs/GiroDF14,DBLP:journals/aamas/SeukenZ08}.

Technically, pMDPs label their transitions with polynomials over a finite set of parameters. Fixing all parameter values in a pMDP yields an MDP. 
The synthesis problem considered in this article asks to find parameter values
such that the induced MDPs satisfy reachability constraints.  Such reachability
constraints state that the probability --- under some/all possible ways to resolve
non-determinism in the MDP --- to reach a target state is (strictly) above or
below a threshold.
A sample synthesis problem is thus: ``Are there parameter values such
that for all possible ways to resolve the non-determinism, the probability to
reach a target state exceeds $\nicefrac{1}{2}$?''
Variants of the synthesis problem are obtained by varying the reachability constraints, and the domain of the parameter values. %Details are given in Sect.~\ref{sec:landscape}.
%Furthermore, an important subclass restricts parameter values to those which induce an MDP with the same topology (graph-preserving).
Parameter synthesis is supported by
the model checkers PRISM~\cite{DBLP:conf/cav/KwiatkowskaNP11} and Storm~\cite{DBLP:conf/cav/DehnertJK017}, and dedicated
tools PARAM~\cite{param_sttt} and
PROPhESY~\cite{DBLP:conf/cav/DehnertJJCVBKA15}. 
The complexity of the decision problems corresponding to parameter synthesis is mostly open.

This article significantly extends complexity results for parameter
synthesis in pMCs and pMDPs.
Tables~\ref{tab:qualcomplexity} and~\ref{tab:complexity} on~pages
\pageref{tab:qualcomplexity} and
\pageref{tab:complexity} give an
overview of new results:  
Most prominently, we establish completeness for the \emph{Existential Theory of the Reals} (ETR) of reachability
problems for pMCs with non-strict comparison operators, and
\NP-hardness for pMCs with strict comparison operators.  For pMDPs with
universal
non-determinism, it establishes ETR-completeness for any comparison operator.
For existential non-determinism, the synthesis problems are mostly equivalent to their pMC counterparts.  
When considering
pMDPs with a fixed number of variables, we establish \NP~upper bounds
for parameter synthesis under existential or universal non-determinism.
These results are partially based on properties of pMDPs scattered over earlier
works (see below), and use a strong connection between polynomial inequalities and parameter synthesis.

Finally, pMDPs are interesting generalisations of other models: Most importantly,
\cite{DBLP:conf/uai/Junges0WQWK018} shows that parameter synthesis in pMCs is
equivalent to the synthesis of finite-state controllers (with a-priori fixed
bounds) of partially observable MDPs (POMDPs)~\cite{DBLP:books/daglib/0023820}
under reachability constraints.  Thus, as a side product we improve complexity
bounds~\cite{DBLP:journals/toct/VlassisLB12,DBLP:conf/aaai/ChatterjeeCD16} for
(a-priori fixed) memory bounded strategies in POMDPs.

\subsection*{Related work}
Various results in this article extend work by Chonev~\cite{DBLP:conf/rp/Chonev19}, who studied augmented interval Markov chains, a model that coincides with pMCs. Our work also builds upon results by Hutschenreiter \emph{et
al.}~\cite{baiercomplexity}, in particular upon the result that pMCs with an
a-priori fixed number of parameters can be checked in \Ptime. Furthermore, they study
the complexity of PCTL model checking of pMCs.  The complexity of finite-state
controller synthesis in POMDPs has been studied in~\cite{DBLP:journals/toct/VlassisLB12,DBLP:conf/aaai/ChatterjeeCD16}.
Some of the proofs for ETR-completeness presented here reuse ideas from~\cite{DBLP:journals/mst/SchaeferS17}.

Methods (and implementations) to analyse pMCs by computing their characteristic
\emph{solution function} are considered in \
\cite{DBLP:conf/ictac/Daws04,param_sttt,DBLP:conf/cav/DehnertJJCVBKA15,baiercomplexity,DBLP:journals/tse/FilieriTG16,DBLP:conf/qest/JansenCVWAKB14,DBLP:journals/ai/DelgadoSB11,DBLP:conf/atva/GainerHS18}.
Sampling-based approaches to find feasible (i.e., satisfying) instantiations 
considered by~\cite{DBLP:conf/nfm/HahnHZ11,DBLP:conf/tase/ChenHHKQ013},
while~\cite{DBLP:conf/tacas/BartocciGKRS11,DBLP:conf/atva/CubuktepeJJKT18}
utilise optimisation methods. 
Finally,~\cite{DBLP:conf/atva/QuatmannD0JK16} presents a method to prove the absence of solutions in pMDPs by iteratively considering simple stochastic games~\cite{DBLP:books/daglib/0074447}.
Some other works on Markov models
with structurally equivalent yet parameterised dynamics
include~\cite{chatterjee12,solan03,DBLP:conf/concur/ChenFRS14,DBLP:journals/acta/CeskaDPKB17}.
Parameter synthesis with statistical guarantees has been explored in, e.g., \cite{DBLP:conf/tacas/BortolussiS18}. 
Novel methods for parametric models under Boolean parameters (i.e., parameter values are restricted to zero or one) have recently been presented in \cite{unpublished:cegis,DBLP:conf/tacas/CeskaJJK19}.
Further work on parameter synthesis in Markov models has been surveyed in~\cite{prophesy_journal}.

\subsection*{Contributions}
The main contribution of this paper is a concise and complete discussion of
the complexity landscape for parameter synthesis in pMCs and pMDPs, as
summarised in Tables~\ref{tab:qualcomplexity} and~\ref{tab:complexity}.  In
particular, we consider a set of decision problems that ask whether there
exists a parameter valuation of a particular type such that, if we substitute a
pMDP (or pMC) with this valuation, the resulting MDP (or MC) satisfies a
quantitative or qualitative reachability property.

The tables contain some known results (mentioned above) that are now part of a
larger picture, but they also contain various new results.  We consider the
following theorems central contributions.
\begin{itemize}
\item \emph{Parameter synthesis in pMCs is \ETR-complete for non-strict
    relations regarding quantitative reachability
  (Theorem~\ref{thm:etr:pmcsarehard}).} Conceptually, this means that
  parameter synthesis is as hard as answering whether a multivariate
  polynomial has a root. Interestingly, this result can be established using
  very simple pMCs. 	
\item \emph{Parameter synthesis in pMDPs is \ETR-complete for any relation
    regarding quantitative reachability (Theorem~\ref{thm:etr:mdpsarehard}).} This result
  is a straightforward adaption of deep results about the existential theory
  of the reals.
\item \emph{The results above are independent of whether or not the parameter
  valuations are graph-preserving}. Graph-preserving valuations simplify
  matters as they allow for stronger continuity assumptions, and are therefore
  standard in tool support for parameter synthesis. The results above show
  that they provide, from a complexity point of view, no benefit.  
\item \emph{Parameter synthesis for qualitative reachability is \NP-hard in
  general (Theorem~\ref{thm:gen-upper-bnd}) but various special cases can be
  decided in polynomial time (Theorem~\ref{thm:easy-cases}).} Results for pMCs
  and pMDPs coincide. To the best of our knowledge, the results cover \emph{all}
  classes considered in the literature on parameter synthesis in pMDPs. 
\item \emph{For any fixed number of parameters, pMDP parameter synthesis is in
  \NP (Theorem~\ref{thm:fp_ea_reach_in_np}).} We would like to stress that this result is non-trivial, as parameter
  values may be real-valued.
\end{itemize}
%These contributions are accompanied by concisely showing upper bounds with adequate encodings, and relating various subclasses of pMDPs.
%
%The table also contains some open questions. The most interesting one seems to be a lower bound for parameter synthesis in pMDPs with a single parameter and quantitative reachability. 
%Another question is precise complexity class of parameter synthesis in pMCs with arbitrarily many parameters and strict bounds on the reachability probability. 
%
%
The presented results extend some results in
\cite{DBLP:conf/concur/WinklerJPK19} by providing examples, full proofs, and
novel results on qualitative variants of the reachability problem.  The
presentation is partially based on \cite{Jun20}.

%% file: 2_preliminaries.tex
\label{sec:prelims}
We assume familiarity with basic graph, automata, and complexity theory.
Below, we present our notation for the theory of the reals and Markov models.

\subsection{Existential theory of the reals}
The \emph{first-order theory of the reals} is the set of all valid sentences
in the first-order language $(\mathbb{R},+,\cdot,0,1,<)$. The existential
theory of the reals (written ETR, for short) restricts the language to
(purely) existentially quantified sentences. The complexity of deciding
membership, i.e.\ whether a sentence is (true) in the theory of the reals, is
in \PSPACE~\cite{DBLP:conf/stoc/Canny88} and \NP-hard. A careful analysis of
its complexity is given in~\cite{DBLP:journals/jsc/Renegar92}. In particular,
deciding membership for sentences with an a-priori fixed upper bound on the
number of variables is in polynomial time. We write \ETR~to denote the complexity
class~\cite{DBLP:journals/mst/SchaeferS17} of problems with a polynomial-time
many-one reduction to deciding membership in the existential theory of the
reals.

\subsection{Markov models}

Markov models are stochastic state models that exhibit the Markov property:
Given any current state, the probability distribution describing the next
state is independent of previous states. In this work by Markov models
we mean \emph{discrete-time Markov chains} (MCs)~\cite{Chu67,KS76,Hag02}
and \emph{Markov decision processes} (MDPs)~\cite{How71,Put94,Kal11}.
We mostly follow the notation from~\cite{BK08}.

\paragraph{Markov decision processes and chains} A \emph{Markov decision process} (MDP)
is a tuple $\mdp \coloneqq (S, \init, \Act, \prob)$ where $S$ is a finite set of
\emph{states}, $\init \in S$ is an \emph{initial state}, $\Act$ is a finite
set of \emph{actions}, and  $\prob \colon S\times \Act \times S \nrightarrow [0,1]$ is
a partial \emph{transition probability} function such that for all $s \in S$, $\act \in \Act$ we either have that $\sum_{s' \in S} \prob(s,\act, s') = 1$ or $\prob(s,\act, s') = \bot$ (undefined) for all $s'$.
%
%We denote $\prob(s, \act)(s')$ with $\prob(s,\act,s')$.
Let ${\EnAct(s) \coloneqq \{ \act \in \Act \mid \forall s' : \prob(s,\act,s') \neq \bot \}}$ denote the
available actions in state $s$.
Without loss of generality, we assume that $|\EnAct(s)| \geq 1$ for all $s \in S$.
Furthermore, we refer to a transition $P(s,\act,s) = 1$ as a self-loop.

A (\emph{discrete-time}) \emph{Markov chain} (MC) is an MDP such
that $|\EnAct(s)| = 1$ for all states $s \in S$.  We may denote an MC as tuple
$\dtmc \coloneqq (S, \init, \prob)$ with $S, \init$ as for MDPs and
\emph{transition probability} function $P\colon S \times S \rightarrow
[0,1]$.
%%
%For Markov chains, we denote $\prob(s)(s')$ with $\prob(s, s')$.

\paragraph{Paths and sets thereof}
We fix an MDP $\mdp \coloneqq (S, \init, \Act, \prob)$.  A \emph{path}
is an (in)finite sequence $\pi \coloneqq
s_0\xrightarrow{\act_0}s_1\xrightarrow{\act_1}\dots$, where $s_i\in S$,
$\act_i\in\EnAct(s_i)$, and $\prob(s_i,\act_i,s_{i+1})\neq 0$ for all
$i\in\mathbb{N}$.  The set $\pathset{\mdp}$ of paths in $\mdp$ is the union of
finite paths $\pathsfin{\mdp}$ and infinite paths $\pathsinf{\mdp}$. The
notions of paths carry over to MCs (actions are omitted).

For finite $\pi = s_0\xrightarrow{\act_0}s_1\xrightarrow{\act_1} \dots s_n$,
we furthermore define the \emph{length} $|\pi| \coloneqq n+1$ of $\pi$ and
$\last{\pi} \coloneqq s_n$. For
infinite paths, we set $|\pi| \coloneqq \infty$.  For any path $\pi =
s_0\xrightarrow{\act_0}s_1\dots$  we set $\state{i}{\pi} \coloneqq s_i$.  A
path $\pi$ \emph{visits} a state $s$, if there is some $i \in \NN$ such that
$\state{i}{\pi} = s$.

We mostly consider paths between or from fixed states. Let $S' \subseteq S$
be a subset of the states and $s \in S$ a state.  The set
\(
  \pathset{\mdp}(S') \coloneqq \{ \pi \in \pathset{\mdp} \mid \first{\pi} \in
  S'\}
\)
contains all paths starting in some $s \in S'$.  We simplify the notation of
the set $\pathset{\mdp}(\{ s \})$ to $\pathset{\mdp}(s)$.  Analogously,
$\pathsfin{\mdp}(s)$, $\pathsinf{\mdp}(s)$, $\pathsfin{\mdp}(S')$,
$\pathsinf{\mdp}(S')$ define the (in)finite  paths starting in some $s'\in
S'$, respectively. 
The set
\[
  \pathset{\mdp}(\eventually{S'}) \coloneqq \{ \pi \in \pathsfin{\mdp} \mid
  \last{\pi} \in S' \land \forall i < |\pi|:\state{i}{\pi} \not\in S' \}
\]
contains all finite paths that end in $S'$ with no proper prefix visiting a
state $s \in S'$.  Again, we simplify notation of
$\pathset{\mdp}(\eventually{\{s\}})$ to $\pathset{\mdp}(\eventually{s})$.
Similarly, for any \emph{horizon} $h \in \NN$, the set \[
\pathset{\mdp}(\eventually[h]{S'}) \coloneqq \{ \pi \in \pathsfin{\mdp} \mid
\last{\pi} \in S' \land |\pi| \leq h \land \forall i < |\pi|:\state{i}{\pi}
\not\in S' \} \] contains all paths of length at most $h$ that end in $S'$
with no proper prefix visiting a state $s \in S'$.  For subsets $S',T \subseteq
S$, the sets 
\[ 
  \pathset{\mdp}(S', \eventually{T}) \coloneqq \pathset{\mdp}(S') \cap
  \pathset{\mdp}(\eventually{T})
\]
denote the paths of states starting in $S'$ and reaching $T$.
We simplify notation for singleton sets as above.

\paragraph{Underlying graphs}
MDPs may be considered as annotated graphs and this perspective helps in describing many operations on MDPs.
For an MDP $\mdp = (S, \init, \Act, \prob)$, the \emph{underlying digraph} of $\mdp$
is $ \graph{\mdp} \coloneqq (S, \edges{\mdp})$ with 
\[
  \edges{\mdp} \coloneqq \{ (s,s') \mid \exists \act \in \EnAct(s): P(s,\act,s') \neq 0  \}.
\]
This definition allows to lift various definitions from graphs to MDPs. 

\paragraph{MDP strategies and induced chains}
To define a probability measure over paths, action choices have to be resolved. Actions are resolved using \emph{strategies}.
A \emph{strategy} for an MDP $\mdp = (S, \init, \Act, \prob)$ is a
(measurable) function \( \sched\colon\pathsfin{\mdp}\to\Distr(\Act)
\text{ such that }\supp(\sched(\pi)) \subseteq \EnAct(\last{\pi})\text{
for all }\pi\in \pathsfin{\mdp}.\)
\begin{itemize}
  \item A strategy is \emph{memoryless} if for all $\pi, \pi' \in
    \pathsfin{\mdp}$ we have \(\last{\pi} = \last{\pi'} \implies 
    \sched(\pi) = \sched(\pi')\).
  \item A strategy is \emph{deterministic} if for all $\pi \in
    \pathsfin{\mdp}$ we have \( |\supp(\sched(\pi))|	 =
    1\).
\end{itemize}
The set of all strategies of $\mdp$ is $\Sched{\mdp}$, the set of all
memoryless strategies is $\MSched{\mdp}$, and the set of all deterministic and
memoryless strategies is $\DMSched{\mdp}$.
Notice that $|\DMSched{\mdp}| \leq |\Act|^{S}$, so in particular there are only finitely many strategies that are deterministic and memoryless.
We may use the function signature $\sched\colon\pathsfin{\mdp}\to\Act$ for deterministic strategies,
and $\sched\colon S\to\Distr(\Act)$ for memoryless strategies.

Let $\sched \in \Sched{\mdp}$ be a strategy.
The \emph{induced MC} of $\mdp$ and $\sched$ is given by
$\induced{\mdp}{\sched} \coloneqq
(\pathsfin{\mdp},\init,\induced{\prob}{\sched})$ where 
\[
  \induced{\prob}{\sched}(\pi,\pi') \coloneqq
  \begin{cases}
      \prob(\last{\pi},\act, s') \cdot \sched(\pi)(\act) & \text{if }\pi' = \pi\xrightarrow{\act} s',  \\
      0 & \text{otherwise.}
    \end{cases}
\]
For memoryless strategies $\sched$, the MC $\induced{\mdp}{\sched}$ may be
identified with the finite MC $\induced{\mdp}{\sched}' \coloneqq (S, \init,
\induced{\prob}{\sched}')$ where \[\induced{\prob}{\sched}(s,s') \coloneqq
\sum_{\act \in \EnAct(s)} \prob(s,\act, s') \cdot \sched(s)(\act).\] Formally,
$\induced{\mdp}{\sched}$ is \emph{probabilistic bisimilar} (cf.~\cite[Sect.\
10.4.2]{BK08}) to $\induced{\mdp}{\sched}'$. In particular, all reachability
probabilities are preserved.  This equivalence justifies using
$\induced{\mdp}{\sched}'$ as redefinition of $\induced{\mdp}{\sched}$ for
memoryless strategies.

\paragraph{Reachability probabilities}
A probability measure $\prsym^\dtmc \colon \pathsfin{\dtmc}\to [0,1]$ for
finite paths $\pi = s_0s_1 \dots s_n$ is given by the product of transition
probabilities, referred to as the \emph{mass of the path}:
\(
  \prpath{\dtmc}{\pi} \coloneqq \prod_{i=0}^{n-1}\prob(s_i,s_{i+1}).
\)
The \emph{unique probability measure} for infinite paths $\prsym^\dtmc \colon
\pathsinf{\dtmc}\to [0,1]$ is defined by the usual cylinder set construction,
see~\cite{BK08} for details.

We define the \emph{reachability probability}
$\pr[s]{\dtmc}{}{\eventually\target}$ for reaching $\target$ from state $s$ as
follows:
\[
 \pr[s]{\dtmc}{}{\eventually{\target}} \coloneqq \sum_{\pi \in
 {\pathset{\dtmc}(s,\eventually T)}} \prpath{\dtmc}{\pi}.
\]
The reachability probability $\pr{\dtmc}{}{\eventually\target}$ for
reaching $\target$ in the MC $\dtmc$ is then defined as the reachability probability from
the initial state.

\paragraph{Optimal strategies for reachability}
A classical result we use in this work is the fact that
deterministic and memoryless strategies suffice in order to optimize
reachability probabilities in MDPs~\cite{Put94}.
\begin{proposition}\label{pro:detmemless-suff}
  For any given MDP $\mdp$, it holds that
	\[
    \sup_{\sched \in \Sched{\mdp}}
    \pr{\induced{\mdp}{\sched}}{}{\eventually{\target}} = \sup_{\sched \in
    \DMSched{\mdp}} \pr{\induced{\mdp}{\sched}}{}{\eventually{\target}}
  \]
  and thus since $\DMSched{\mdp}$ is a finite set, the suprema may be replaced
  by maxima. An analogous statement holds for infima and minima.
\end{proposition}

\paragraph{Computing reachability values}
A final result which we will repeatedly make use of is the well-known fact
that minimal and maximal reachability probabilities are polynomial-time
computable. 
\begin{proposition}\label{pro:poly-time-MDPs}
  For any given MDP $\mdp$, the values $\max_{\sched \in \DMSched{\mdp}}
  \pr{\induced{\mdp}{\sched}}{}{\eventually{\target}}$ and $\min_{\sched \in
  \DMSched{\mdp}} \pr{\induced{\mdp}{\sched}}{}{\eventually{\target}}$ can
  be computed in polynomial time.
\end{proposition}
This follows from a straightforward encoding of the value into a linear
program. We refer the reader to~\cite{BK08,Put94} for a proof.

%% file: 3_pmdp.tex
In this section we introduce parametric MDPs and parametric MCs and provide
examples of what can be modelled by them.
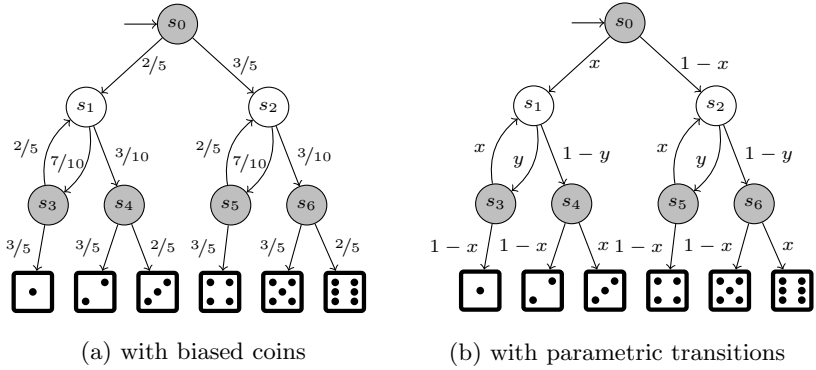
\begin{figure}
\centering
	\begin{subfigure}{0.45\textwidth}
		\input{figures/parametric_kydie_unfair}
		\subcaption{with biased coins}
		\label{fig:parametric:pkydiei}
	\end{subfigure}
	\begin{subfigure}{0.45\textwidth}
		\input{figures/parametric_kydie_model}
		\subcaption{with parametric transitions}
		\label{fig:parametric:pkydiepmc}
	\end{subfigure}
	\caption{Two variants of the Knuth-Yao die}
\end{figure}
\begin{example}
\label{ex:parametric:kydie}
The Knuth-Yao algorithm~\cite{KY76} uses repeated coin flips to model a six-sided die.
It uses a fair coin to obtain each possible outcome (`one', `two', ..., `six') with probability $\nicefrac{1}{6}$.
Figure~\ref{fig:parametric:pkydiei} depicts an MC of a variant in which two \emph{unfair} coins are flipped in an alternating fashion.
Flipping the coins yields \emph{heads} with probability $\nicefrac{2}{5}$ (gray states) or $\nicefrac{7}{10}$ (white states), respectively.
Accordingly, the probability of \emph{tails} is $\nicefrac{3}{5}$ and $\nicefrac{3}{10}$, respectively.
The event of throwing a `two' corresponds to reaching the state~\drawdie[scale=0.7]{2} in the MC.
Assume now a specification requiring the probability to obtain `two' to be larger than $\nicefrac{3}{20}$. 
Knuth-Yao's original algorithm satisfies this property as using a fair coin results in $\nicefrac{1}{6}$ as probability to reach \drawdie[scale=0.7]{2}.
The biased model however, does not satisfy the property; in fact, \drawdie[scale=0.7]{2} is reached with probability $\nicefrac{1}{10}$. 
We may now ask ourselves: how unfair may these coins be and still satisfy the
property? Vice versa, we can ask: Assuming that the probability to throw
\emph{heads} is between $\nicefrac{2}{5}$ and $\nicefrac{3}{5}$, does the
property hold for all admissible probabilities of throwing \emph{heads}?
\end{example}
When analysing parameterised MCs, we consider reachability
properties on some or all instantiations.
In the example above, we asked whether all `almost fair' coins satisfy
reaching \drawdie[scale=0.7]{2} with at least some given probability.
%%
%\begin{figure}
%\centering
%\begin{subfigure}[b]{0.62\textwidth}
%\centering
%\input{figures/parametric_rps}
%\subcaption{The parametric MDP for modified rock-paper-scissors.}
%\label{fig:parametric:rpspmdp}	
%\end{subfigure}
%\begin{subtable}[b]{0.37\textwidth}
%\scriptsize
%\centering
%\begin{tabular}{ll|lll}
%                  &          & \multicolumn{3}{c}{We play}    \\
%                  &          & R     & P    & S \\ \hline
%\multirow{3}{*}{\shortstack[l]{Eve\\plays}} & R     & d     & w   & l \\
%                  & P    & l & d     & w   \\
%                  & S & w   & l & d    
%\end{tabular}
%\subcaption{Winning a round of rock-paper-scissors}
%\label{tab:rps}
%\end{subtable}
%\caption{Playing (a slightly modified) rock-paper-scissors.}
%\end{figure}
%%
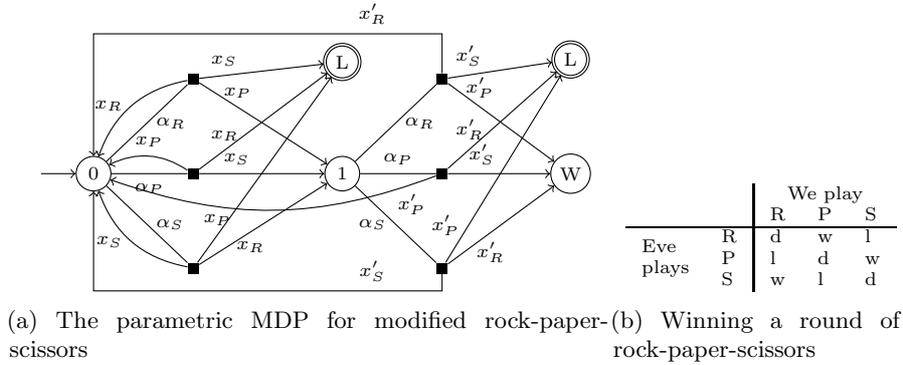
\begin{figure}
	\centering
	\subcaptionbox{The parametric MDP for modified rock-paper-scissors \label{fig:parametric:rpspmdp}}{
		\input{figures/parametric_rps}
	}
	\subcaptionbox{Winning a round of rock-paper-scissors \label{tab:rps}}{
		\scriptsize
		\begin{tabular}{ll|lll}
			&          & \multicolumn{3}{c}{We play}    \\
			&          & R     & P    & S \\ \hline
			\multirow{3}{*}{\shortstack[l]{Eve\\plays}} & R     & d     & w   & l \\
			& P    & l & d     & w   \\
			& S & w   & l & d    
		\end{tabular}
	}
	\caption{Playing (a slightly modified) rock-paper-scissors.}
\end{figure}
\begin{example}
\label{ex:parametric:rps}
	We consider a slightly modified variant of rock-paper-scissors, in which we play against Eve. Eve selects behind her back either Rock (R), Paper (P) or Scissors (S). Now, without knowing what she selected, we have to make a (randomised) decision. 
	Then, we win the round as usual, in accordance with Table~\ref{tab:rps}.
	We are interested in minimising the probability that we ever loose a round. Luckily, we thus either (1) win by playing (draws) forever, or (2) Eve gives up if we win twice in a row.  
	We model this protocol by the pMDP in Figure~\ref{fig:parametric:rpspmdp}.
	Initially, Eve selects either (R), (P) or (S). Then, we select with some probability $x_{P}$ \emph{paper}, $x_R$ \emph{rock}, and $x_S$ \emph{scissors}. This probability is independent of the selection of Eve.
	If we lose, we go to a target state  $L$ (duplicated to avoid clutter). 
	Otherwise, if we win, we move to the next round, 
	and if we draw, we restart the game. 
	In the second round, we may choose a different distribution over the actions. 
	The distribution is represented by $x'_P, x'_R, x'_S$.
	After two wins, we reach the sink state~$W$.
	
	We want to know how we should choose the actions to avoid reaching $L$. 
	On a technical level, this questions corresponds to asking for the right values for $x_{R}, x_{P}, x_{S}$.
	 In the worst-case, Eve has some transcendental powers and always counters optimally.  What is the best strategy for us, i.e., what is the minimal probability of reaching the target state $L$? And how may we randomise if we merely want to ensure that we reach $L$ with a probability less than 90\%? Clearly, we have to randomise, otherwise Eve will easily counter our moves.
\end{example}

\noindent We now formally define parametric MDPs and their instantiations.

\subsection{Fundamentals}
Parametric MDPs extend MDPs
with a set $\pars$ of parameters and an adapted
transition relation: Probabilities are no longer expressed by values from
$[0,1]$, but by rational functions over $\pars$\footnote{For technical
reasons, we  exclude non-rational probabilities which are allowed in MDPs. The field of rational functions (polynomials) with coefficients in $\QQ$ is denoted $\QQ(\pars)$ ($\QQ[\pars]$, respectively).}.
\begin{definition}[pMDP~\cite{DBLP:conf/nfm/HahnHZ11}]
  A \emph{parametric Markov Decision Process} (\emph{pMDP}) $\pmdp$ is a tuple $\pmdptuple$
  with a  finite set $S$ of \emph{states}, an initial state $\init \in S$, a finite set $\Act$ of \emph{actions}, a finite set $\pars$ of
  \emph{parameters}, and a parametric probabilistic transition $\pprob \colon S
  \times \Act \times S \rightarrow \QQ(\pars)$.
\end{definition}
\noindent We limit ourselves to \emph{polynomial pMDPs}:
A pMDP $\pmdp$ is \emph{polynomial} if
\[
  \forall s, s' \in S, \act \in
  \EnAct(s): \pprob(s,\act,s') \in \QQ[\pars].
\]
The rationale is twofold. First, instantiating rational functions yields
undefined values when the denominator becomes zero.  That induces an
additional case distinction which is undesirable for a concise presentation of
the theory.  Furthermore, we want to use transition probabilities in the ETR constraints.  However, the ETR does not
`natively' support rational functions\footnote{It would require, e.g., defining
whether $\nicefrac{0}{0} > 0$.}.

As pMDPs extend (finite) MDPs, the concepts of paths, graphs, and strategies
as in Section~\ref{sec:prelims} carry over naturally.  Likewise,
parametric MCs are obtained as an extension to MCs, and are a special case of
pMDPs.
\begin{definition}[pMC]
	A \emph{parametric Markov chain} (\emph{pMC}) $\pmc$ is a pMDP \\ $\pmdptuple$ such that $|\EnAct(s)| = 1$ for all states $s \in S$. We identify the parametric probabilistic transition of $\pmc$ with a function $\pprob \colon S \times S \rightarrow \QQ(\pars)$.
%  A \emph{parametric Markov chain} (\emph{pMC}, for short) $\pmc$ is a tuple $\pmctuple$
%  with a  finite set $S$ of \emph{states}, an initial state $\init \in S$, a finite set $\pars$ of
%  \emph{parameters}, and a parametric probabilistic transition $\pprob \colon S
%\times S \rightarrow \QQ(\pars)$.
\end{definition}
%%
%Our definition of pMCs is based on the definitions in \cite{DBLP:conf/ictac/Daws04,DBLP:journals/fac/LanotteMT07} and the pMDP definition originates from \cite{DBLP:conf/nfm/HahnHZ11}.
%%
\begin{example}
  Figure~\ref{fig:parametric:pkydiepmc} depicts a parametric version of the
  biased Knuth-Yao die from Example~\ref{ex:parametric:kydie}.  It has
  parameters $\pars = \{x,y\}$, where $x$ is the probability of outcome
  \emph{heads} in grey states and $y$ the same for white states. 	The
  probability for \emph{tails} is then $1{-}x$ and $1{-}y$, respectively.
  Figure~\ref{fig:parametric:rpspmdp} depicts the pMDP for our
  rock-paper-scissors variant. The parameters are $\pars = \{ x_R, x_P, x_S,
  x'_R, x'_P, x'_S \}$.
\end{example}
Parameters are thus variables that may be substituted by concrete values. Not
all valuations are meaningful in the context of pMDPs.  We are only interested
in valuations for pMDPs that yield MDPs. Such \emph{valuations} are called
well-defined.
\begin{definition}[Well-defined valuation]
\label{def:parametric:welldefined}
	Let $\pmdp$ be a pMDP with parameters $\pars$.
	A valuation $\val \colon \pars \to \RR$ is \emph{well-defined} for $\pmdp$ if:
	\begin{itemize}
		\item probabilities are non-negative, i.e., $\pprob(s,\act,s')[\val] \geq 0$ for all $s,s' \in S, \act \in \EnAct(s)$.
		\item outgoing probabilities induce distributions, i.e., $\sum_{s' \in S} \pprob(s,\act,s')[\val] = 1$ for all $s \in S, \act \in \EnAct(s)$.
	\end{itemize}
	The set $\wdval{\pmdp}$ consists of all well-defined valuations for $\pmdp$.
\end{definition}
\noindent Well-defined valuations of $\pars$ for $\pmdp$ are just called \emph{valuations for $\pmdp$}.
\begin{figure}
\centering
	\input{figures/parametric_unrealisable}
	\caption{An unrealisable pMC}
	\label{fig:parametric:unrealisable}
\end{figure}
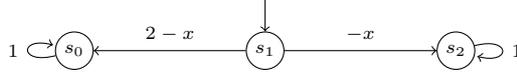
\begin{example}
\label{ex:parametric:welldefined}
The well-defined valuations for the pMDP in Figure~\ref{fig:parametric:rpspmdp} are \begin{align*} \wdval{\pmdp} = \{ \val \in \Val \mid~ & \val(x_P) + \val(x_R) + \val(x_S) = 1, 
    \\ & \val(x_P), \val(x_R), \val(x_S) \geq 0, 
    \\ & \val(x'_P) + \val(x'_R) + \val(x'_S) = 1, 
    \\ & \val(x'_P), \val(x'_R), \val(x'_S) \geq 0 \}. 
    \end{align*}
A valuation $\{ x_R, x_P, x_S, x'_R, x'_P, x'_S \mapsto \nicefrac{2}{3}  \}$ is not well-defined, as the sum of $\val(x_R), \val(x_P), \val(x_S)$ exceeds one.
Some pMDPs do not have any well-defined valuation, e.g., the pMC in Figure~\ref{fig:parametric:unrealisable}. It can be readily checked that no $x$ satisfies $2-x + (-x) = 1$ and $-x \geq 0$.
\end{example}
\begin{definition}[Realisable]
A pMDP $\pmdp$ is \emph{realisable}, if $\wdval{\pmdp}\neq \emptyset$.	
\end{definition}

Let $\pmdp$ be a pMDP and $\val \in \wdval{\pmdp}$ a valuation.
The \emph{instantiation} of $\pmdp$ with $\val$ is the MDP
$\instantiated{\pmdp}{\val} \coloneqq (S, \init, \Act, \prob)$ with
\[
  \prob(s,\act,s') \coloneqq \pprob(s,\act,s')[\val]\quad\text{for all }s, s'
  \in S, \act \in \EnAct(s).
\]

\begin{example}
  Reconsider the pMC in Figure~\ref{fig:parametric:pkydiepmc}. Observe that
  the valuation $\val$ with ${\val(x) = \nicefrac{2}{5}}$ and ${\val(y) =
  \nicefrac{7}{10}}$ is well-defined and yields the MC in
  Figure~\ref{fig:parametric:pkydiei}.
\end{example}

The notion of strategies carries over from MDPs to pMDPs and thus induced MCs
carry over to induced pMCs.

It is helpful and natural to consider a pMDP $\pmdp$ as a \emph{generator} for
an, in general, uncountable set $\generator{\pmdp}$ of instantiated MDPs.
\begin{definition}[Generator]
  The generator of pMDP $\pmdp$ is the set 
    \[ \generator{\pmdp} \coloneqq \{ \instantiated{\pmdp}{\val} \mid \val \in \wdval{\pmdp} \}. \]
    Let  $\region \subseteq \wdval{\pmdp}$ be a set of valuations for $\pmdp$. We define the set 
	\[ \generator[\region]{\pmdp} \coloneqq \{ \instantiated{\pmdp}{\val} \mid \val \in \region \}. \]
\end{definition}
Thus, $\generator{\pmdp} = \generator[\wdval{\pmdp}]{\pmdp}$. 
For this set, and sets yet to be introduced, we  often omit everything but the superscript and write, e.g., $\generator[\mathrm{wd}]{\pmdp}$.

\subsection{Solution functions of parametric models}
\label{sec:parametric:measures}
Each well-defined parameter valuation yields an instantiation for which the
measures are defined. Hence, we map valuations to reachability probabilities.
\emph{Solution functions} capture this mapping.
\begin{definition}[Solution function]
  For a pMC $\pmc$ and a state $s$, 
  \label{def:solutionfunction}
  let the \emph{(probability) solution function} \(
  \prsol[s]{\target}{\pmc} \colon  \wdval{\pmc} \rightarrow [0,1]\)
  be
  \[
    \prsol[s]{\target}{\pmc}(\val) \coloneqq
    \pr[s]{\instantiated{\pmc}{\val}}{}{\eventually\target}.
  \]
%  
%For a pMDP
%$\pmdp$, let \( \minprsol[s]{\target}{\pmdp} \colon  \wdval{\pmdp} \rightarrow
%[0,1]\) be defined as
%\[
%\minprsol[s]{\target}{\pmdp}(\val) \coloneqq \min_{\sched \in \DMSched{\pmdp}} \prsol[s]{\target}{\induced{\pmdp}{\sched}}(\val) = 
%  \min_{\sched \in \DMSched{\pmdp}}
%  \pr[s]{\instantiated{\pmdp}{\val}}{\sched}{\eventually\target}.
%\]\todo{JP: notion of minsol not used in paper, only example 6}
\end{definition}
We omit the initial state whenever possible and omit $\pmdp$ and $\target$
whenever they are clear from the context.
\begin{figure}
\centering
\begin{subfigure}[b]{0.45\textwidth}
\centering
	\input{figures/parametric_mini_pmc}
	\caption{A tiny pMC}
	\label{fig:parametric:tinypmc}
\end{subfigure}
\begin{subfigure}[b]{0.45\textwidth}
\centering
	\input{figures/parametric_mini_pmdp}
	\caption{A tiny pMDP}
	\label{fig:parametric:tinypmdp}
\end{subfigure}
\caption{Two small, acyclic models}	
\label{fig:parametric:tiny}
\end{figure}
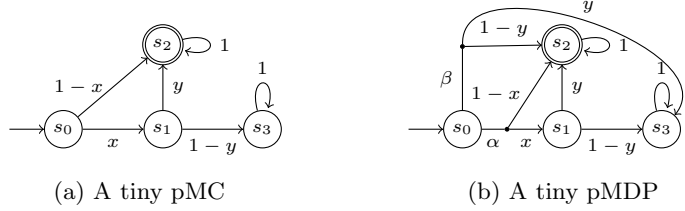

\begin{example}
\label{ex:parametric:solfuncs}
	Consider the pMC $\pmc$ in Figure~\ref{fig:parametric:tinypmc}. There are two paths to the target state. The probability $f \coloneqq x \cdot y + 1 - x$ to reach the target is the sum over the probabilities over these two paths.
	For any well-defined instantiation $\val$, the probability to reach the target in $\pmc[\val]$ is $f[\val]$.
	Thus, \[ \prsol{\target}{\pmc} = x \cdot y + 1 - x. \]
%	For the pMDP in Figure~\ref{fig:parametric:tinypmdp}, we have (after reformulation) \[ \minprsol{\target}{\pmdp} = \min \{1 - (1-y)x, 1-y \}. \]
%	We observe that for large values for $y$, action $\beta$ minimises the probability to the target, whereas for smaller values of $y$, depending on the value for $x$, $\alpha$ minimises the probability. 
\end{example}

\subsection{Graph-preserving valuations}
Recall that $\generator{\pmdp}$ considers the instantiations that are induced
by well-defined valuations.  Below, we consider a restriction on the
valuations.  In the analysis of parameter-free MDPs, it is often essential to
exploit the topology of the MDP, e.g., when computing zero-states.  In
$\generator{\pmdp}$, not all MDPs have the same topology.  The goal below is
to consider a restriction on the valuations such that all MDPs in
$\generator{\pmdp}$ have the same topology as $\pmdp$.  The topology changes,
if a transition is present in the pMDP but not in its instantiation.
\begin{definition}[Vanishing transitions]
  Let $\pmdp$ be a pMDP with $\val \in \wdval{\pmdp}$. We call a transition
  $(s,\act,s)$ \emph{vanishing} under $\val$ if
	\[
    \pprob(s,\act,s') \neq 0 \text{ and } \pprob(s,\act,s')[\val] = 0.
  \]
  The set $\vanish[\pmdp]{\val}{} \subseteq S \times \Act \times S$ contains
  all vanishing transitions.
\end{definition}
\noindent A valuation  preserves the topology if no transitions vanish, formally:
\begin{definition}[Graph-preserving valuations]
\label{def:parametric:graph_preservation}
  Let $\pmdp$ be a pMDP. A valuation $\val \in \wdval{\pmdp}$ is
  \emph{graph-preserving} if $\vanish[\pmdp]{\val}{} = \emptyset$. The set
  $\gpval{\pmdp}$ contains all graph-preserving valuations.
\end{definition}
\noindent The generator for this class is $\generator[\gpval{\pmdp}]{\pmdp}$, also denoted $\generator[\mathrm{gp}]{\pmdp}$.

\begin{example}
  Let us again consider Figure~\ref{fig:parametric:pkydiepmc}.  If we set
  $\val(x) \coloneqq 0$, then various transitions disappear, in particular the
  transitions from $s_0$ to $s_1$, and from $s_3$ to $s_1$. Thus, any
  valuation with $\val(x) = 0$ is not graph-preserving.
	
  There exist realisable pMCs without graph-preserving instantiations, e.g.,
  any realisable pMC with states $s, s_1, s_2$ such that $\pprob(s,s_1) = x$,
  $\pprob(s,s_2) = x + 1$.
\end{example}

\subsection{Other sets of valuations}
Sets of valuations may have particular characteristics. For example, when all
valuations are graph-preserving, it is a graph-preserving set of valuations.
Slightly weaker, the set is \emph{graph-consistent}, when all valuations
induce the same topology (but not necessarily the topology of the
corresponding pMDP). 
\begin{definition}[Graph-consistent sets of valuations]
  \label{def:parametric:graph_consistent}
  A graph-consistent set $\region$ of valuations is a subset
  of the well-defined valuations such that for all $\val, \val' \in \region$:
  \[
    \vanish[\pmdp]{\val}{} = \vanish[\pmdp]{\val'}{}.
  \]
  It is maximally graph-consistent, if no true superset of $\region$ is
  graph-consistent.
\end{definition}
\begin{example}
  Let us again consider Figure~\ref{fig:parametric:pkydiepmc}. As we have seen
  previously, valuations with $\val(x)$ are not graph-preserving. 
	However, the set
  \[
    \{ \val \in \wdval{\pmc} \mid \val(x) = 0 \land  0 < \val(y) < 1 \}
  \] is maximally graph-consistent. 
\end{example}

Note that inside graph-consistent sets of valuations the sets of vanishing transitions
are invariant. Formally, we have from~\cite{param_sttt} the
following property.
\begin{proposition}\label{pro:encoding:qualitativegraphproperty}
  Let $\pmdp$ be a pMDP with target states $\target$ and $\region$ a
  graph-consistent set of valuations. For all $\mdp, \mdp' \in
  \generator[\region]{\pmdp}$ and
  $\sched \in \DMSched{\pmdp}$:
  \begin{align*}
    & \pr{\mdp}{\sched}{\eventually\target} = 0 \text{ implies }
    \pr{\mdp'}{\sched}{\eventually\target} = 0, \text{ and}\\
    & \pr{\mdp}{\sched}{\eventually\target} = 1 \text{ implies }
    \pr{\mdp'}{\sched}{\eventually\target} = 1.	
  \end{align*}
\end{proposition}
A proof of this claim follows directly from the graph-based algorithms for
\emph{qualitative properties}~\cite{BK08}, that is, whether the maximal or
minimal reachability probabilities are precisely $0$ or $1$.
The same graph-based algorithms suggest that removing transitions does not
increase the number of states from which the reachability probability is
positive:
\begin{lemma}\label{lem:more-better}
  Let $\pmdp$ be a pMDP with target states $\target$.
  For all $\val \in \gpval{\pmdp}$ and all $\sched \in
  \DMSched{\pmdp}$ we have that:
  \[
    \pr{\instantiated{\pmdp}{\val}}{\sched}{\eventually\target} > 0
    \text{ iff there exists some }\mdp \in
    \generator{\pmdp} \text{ s.t. } \pr{\mdp}{\sched}{\eventually\target} > 0.
  \]
\end{lemma}

\paragraph{Boolean valuations} A final class of valuation sets that we
consider is the restriction to $\{0,1\}$. Formally, a
valuation $\val \in \wdval{\pmdp}$ is a \emph{Boolean valuation} if $\val(x) \in \{0,1\}$
for all $x \in \pars$. We write $\pureval{\pmdp}$ for the set of all Boolean
valuations.

\subsection{Problem statement}
The question we address in this article is whether some instantiation of
$\pmdp$ is such that its maximal or minimal probability of  eventually
reaching $T$ compares with  $\lambda \in \{ 0,\nicefrac{1}{2},1 \}$ in some
desired way. In symbols, for $\quant \in \{\exists,\forall\}$, $R \subseteq
\wdval{\pmdp}$, and
$\mathrm{\bowtie} \in \{\leq, <, >, \geq\}$, we consider the decision problem
\[ 
  \exists\;\val \in R,
  \quant\;\sched \in \DMSched{\pmdp}:\;
    \pr{\instantiated{\pmdp}{\val}}{\sched}{\eventually\target} \bowtie \lambda.
\]

\paragraph{Assumptions}
When studying the complexity of reachability problems, we
will mostly focus on \emph{simple pMDPs}.
A pMDP $\pmdp$ is \emph{simple} if
\begin{itemize}
  \item $\mathcal{P}(s,\alpha,s') \in \{x,1-x \mid x \in X\} \cup
    \mathbb{Q}_{\geq 0}$ for all $s,s' \in S$ and $\alpha \in \Act$;
    and
  \item $\sum_{s' \in S} \mathcal{P}(s,\alpha,s') \equiv 1$ for all $s \in S$ and
    $\alpha \in \EnAct(s)$.
\end{itemize}
The well-defined and graph-preserving valuations for simple pMDPs are
$[0,1]^X$ and $(0,1)^X$ respectively.

Note that such pMDPs are essentially a model of sequential parametric
Bernoulli experiments.  The reason we restrict our study to simple pMDPs is to
avoid the complexity being governed by the subproblem of checking whether
there is some well-defined valuation, which in general is an \ETR-hard problem.
\begin{proposition}[From~\cite{DBLP:journals/fac/LanotteMT07}]
\label{pro:aspects:wdetrhard}
	Given a polynomial pMDP with at least two states, determining whether it is
  realisable is \ETR-hard.
\end{proposition}
We give here a simple proof of the claim using a lemma that will be proved in the
sequel.
\begin{proof}
  Consider a pMC with two states and a single transition between them with
  probability $f \in \QQ[\pars]$. The constraints for well-definedness collapse
  to $ f \ceq 1 $, or equivalently $f - 1 \ceq 0$. For multivariate polynomials
  of degree at least four, answering this question is \ETR-hard --- see
  Lemma~\ref{lem:etr:mb4feas}.
\end{proof}

\paragraph{Encoding of the input}
Let $\pmdp$ be a \emph{simple} pMDP with and a set $\target$ of target states.  We analyse the decision problems according to whether the
set $\pars$ of parameters from $\pmdp$ has bounded size --- with a-priori fixed
bound --- or arbitrary size.  It remains to fix an encoding for polynomials with rational coefficients. Henceforth, we assume the exponents of such polynomials are given as binary-encoded integers and the (rational) coefficients as pairs of integers, also encoded in binary.
%Henceforth, we assume the coefficients, exponents and
%constants are all given as binary-encoded integers or --- in the case of rational numbers --- pairs thereof.

%% file: figures/parametric_kydie_unfair.tex
\begin{tikzpicture}[scale=1, die/.style={inner sep=0,outer sep=0}, every node/.style={font=\scriptsize}, st/.style={draw, circle, inner sep=2pt, minimum size=15pt},baseline=(s0)]
%    \draw[white, use as bounding box] (-1.7,-4.5) rectangle (1.7,0.3); 
    
    \node [st,fill=gray!50] (s0) at (0,0) {$s_0$};
    \node [] (leftdummy)  [on grid, left=1.2cm of s0] {};
    \node [] (rightdummy) [on grid, right=1.2cm of s0] {};
    \node [st] (s1) [on grid, below=1.1cm of leftdummy] {$s_1$};
    \node [st] (s2) [on grid, below=1.1cm of rightdummy] {$s_2$};
    \node [st,fill=gray!50] (s3) [on grid, below=1.3cm of s1, xshift=-0.5cm] {$s_3$};
        \node [st,fill=gray!50] (s4) [on grid, below=1.3cm of s1, xshift=0.5cm] {$s_4$};
    \node [st,fill=gray!50] (s5) [on grid, below=1.3cm of s2, xshift=-0.5cm] {$s_5$};
    
    \node [st,fill=gray!50] (s6) [on grid, below=1.3cm of s2, xshift=0.5cm] {$s_6$};

    \node[die, scale=2, below=0.6cm of s3, xshift=-0.1cm] (X1) {\drawdie{1}};
    \node[die, scale=2, right=0.26cm of X1, inner sep=0pt] (X2) {\drawdie{2}};
    \node[die, scale=2, right=0.26cm of X2] (X3) {\drawdie{3}};
    \node[die, scale=2, right=0.26cm of X3] (X4) {\drawdie{4}};
    \node[die, scale=2, right=0.26cm of X4] (X5) {\drawdie{5}};
    \node[die, scale=2, right=0.26cm of X5] (X6) {\drawdie{6}};
    
    \draw ($(s0)-(0.7,0)$) edge[->] (s0);
    \draw (s0) edge[->] node[right] {\scriptsize$\nicefrac{2}{5}$} (s1);
    \draw (s0) edge[->] node[right] {\scriptsize$\nicefrac{3}{5}$} (s2);
    \draw (s1) edge[bend left, ->] node[left,xshift=1mm] {\scriptsize$\nicefrac{7}{10}$} (s3);
    \draw (s1) edge[->] node[right] {\scriptsize$\nicefrac{3}{10}$} (s4);
    \draw (s3) edge[bend left, ->] node[left] {\scriptsize$\nicefrac{2}{5}$} (s1);
    \draw (s3) edge[->] node[left] {\scriptsize$\nicefrac{3}{5}$} (X1);
    \draw (s4) edge[->] node[left] {\scriptsize$\nicefrac{3}{5}$} (X2);
    \draw (s4) edge[->] node[right] {\scriptsize$\nicefrac{2}{5}$} (X3);
    \draw (s2) edge[bend left, ->] node[left,xshift=1mm] {\scriptsize$\nicefrac{7}{10}$} (s5);
    \draw (s2) edge[->] node[right] {\scriptsize$\nicefrac{3}{10}$} (s6);
    \draw (s5) edge[bend left, ->] node[left] {\scriptsize$\nicefrac{2}{5}$} (s2);
    \draw (s5) edge[->] node[left] {\scriptsize$\nicefrac{3}{5}$} (X4);
    \draw (s6) edge[->] node[left] {\scriptsize$\nicefrac{3}{5}$} (X5);
    \draw (s6) edge[->] node[right] {\scriptsize$\nicefrac{2}{5}$} (X6);
    
    \node[draw=white, rectangle, fit=(current bounding box)] {};
\end{tikzpicture}

%% file: figures/parametric_kydie_model.tex
\begin{tikzpicture}[scale=1, die/.style={inner sep=0,outer sep=0}, every node/.style={font=\scriptsize}, st/.style={draw, circle, inner sep=2pt, minimum size=15pt},baseline=(s0)]
%    \draw[white, use as bounding box] (-1.7,-4.5) rectangle (1.7,0.3); 
    
    \node [st,fill=gray!50] (s0) at (0,0) {$s_0$};
    \node [] (leftdummy)  [on grid, left=1.2cm of s0] {};
    \node [] (rightdummy) [on grid, right=1.2cm of s0] {};
    \node [st] (s1) [on grid, below=1.1cm of leftdummy] {$s_1$};
    \node [st] (s2) [on grid, below=1.1cm of rightdummy] {$s_2$};
    \node [st,fill=gray!50] (s3) [on grid, below=1.3cm of s1, xshift=-0.5cm] {$s_3$};
        \node [st,fill=gray!50] (s4) [on grid, below=1.3cm of s1, xshift=0.5cm] {$s_4$};
    \node [st,fill=gray!50] (s5) [on grid, below=1.3cm of s2, xshift=-0.5cm] {$s_5$};
    
    \node [st,fill=gray!50] (s6) [on grid, below=1.3cm of s2, xshift=0.5cm] {$s_6$};

    \node[die, scale=2, below=0.6cm of s3, xshift=-0.1cm] (X1) {\drawdie{1}};
    \node[die, scale=2, right=0.26cm of X1, inner sep=0pt] (X2) {\drawdie{2}};
    \node[die, scale=2, right=0.26cm of X2] (X3) {\drawdie{3}};
    \node[die, scale=2, right=0.26cm of X3] (X4) {\drawdie{4}};
    \node[die, scale=2, right=0.26cm of X4] (X5) {\drawdie{5}};
    \node[die, scale=2, right=0.26cm of X5] (X6) {\drawdie{6}};
    
    \draw ($(s0)-(0.7,0)$) edge[->] (s0);
    \draw (s0) edge[->] node[right] {\scriptsize$x$} (s1);
    \draw (s0) edge[->] node[right] {\scriptsize$1-x$} (s2);
    \draw (s1) edge[bend left, ->] node[left] {\scriptsize$y$} (s3);
    \draw (s1) edge[->] node[right] {\scriptsize$1-y$} (s4);
    \draw (s3) edge[bend left, ->] node[left] {\scriptsize$x$} (s1);
    \draw (s3) edge[->] node[left] {\scriptsize$1-x$} (X1);
    \draw (s4) edge[->] node[left] {\scriptsize$1-x$} (X2);
    \draw (s4) edge[->] node[right] {\scriptsize$x$} (X3);
    \draw (s2) edge[bend left, ->] node[left] {\scriptsize$y$} (s5);
    \draw (s2) edge[->] node[right] {\scriptsize$1-y$} (s6);
    \draw (s5) edge[bend left, ->] node[left] {\scriptsize$x$} (s2);
    \draw (s5) edge[->] node[left] {\scriptsize$1-x$} (X4);
    \draw (s6) edge[->] node[left] {\scriptsize$1-x$} (X5);
    \draw (s6) edge[->] node[right] {\scriptsize$x$} (X6);
    
    \node[draw=white, rectangle, fit=(current bounding box)] {};
\end{tikzpicture}

%% file: figures/parametric_rps.tex
\begin{tikzpicture}[scale=1, every node/.style={scale=1, font=\scriptsize}, st/.style={draw, circle, inner sep=2pt, minimum size=13pt}, initial text=]
\node[initial,st] (c0) {0};
\node[inner sep=2pt, fill=black, right=1cm of c0] (l12) {};
\node[inner sep=2pt, fill=black, above=1.1cm of l12] (l11) {};
\node[inner sep=2pt, fill=black, below=1.1cm of l12] (l13) {};

\node[st, right=2.8cm of c0] (c1) {1};
\node[inner sep=2pt, fill=black, right=1cm of c1] (l22) {};
\node[inner sep=2pt, fill=black, above=1.1cm of l22] (l21) {};
\node[inner sep=2pt, fill=black, below=1.1cm of l22] (l23) {};

\node[st, right=2.5cm of c1] (c2) {W};

\node[st, above=1cm of c1, accepting] (t1) {L};
\node[st, above=1cm of c2, accepting] (t2) {L};

\draw[-] (c0) edge node[right] {$\alpha_R$} (l11);
\draw[-] (c0) edge node[below] {$\alpha_P$} (l12);
\draw[-] (c0) edge node[right] {$\alpha_S$} (l13);

\draw[->] (l11) edge[bend right] node[left] {$x_R$} (c0);
\draw[->] (l11) edge node[above, pos=0.3] {$x_P$} (c1);
\draw[->] (l11) edge node[above, pos=0.2] {$x_S$} (t1);

\draw[->] (l12) edge[bend right] node[above] {$x_P$} (c0);
\draw[->] (l12) edge node[above, pos=0.3] {$x_S$} (c1);
\draw[->] (l12) edge node[above, pos=0.2] {$x_R$} (t1);

\draw[->] (l13) edge[bend left] node[left] {$x_S$} (c0);
\draw[->] (l13) edge node[below, pos=0.4] {$x_R$} (c1);
\draw[->] (l13) edge node[above, pos=0.15] {$x_P$} (t1);

\draw[-] (c1) edge node[right] {$\alpha_R$} (l21);
\draw[-] (c1) edge node[above] {$\alpha_P$} (l22);
\draw[-] (c1) edge node[left] {$\alpha_S$} (l23);

\draw[->] (l21) -- ++(0,.6) -| node[above,pos=0.1] {$x'_R$} (c0);
\draw[->] (l21) edge node[above, pos=0.3] {$x'_P$} (c2);
\draw[->] (l21) edge node[above, pos=0.2] {$x'_S$} (t2);

\draw[->] (l22) edge[bend left=20] node[below,pos=0.07] {$x'_P$} (c0);
\draw[->] (l22) edge node[above, pos=0.33] {$x'_S$} (c2);
\draw[->] (l22) edge node[above, pos=0.2] {$x'_R$} (t2);

\draw[->] (l23) -- ++(0,-.3) -| node[above,pos=0.1] {$x'_S$} (c0);
\draw[->] (l23) edge node[below, pos=0.4] {$x'_R$} (c2);
\draw[->] (l23) edge node[left, pos=0.2] {$x'_P$} (t2);%
\end{tikzpicture}

%% file: figures/parametric_unrealisable.tex
\begin{tikzpicture}[ every node/.style={scale=1, font=\scriptsize}, st/.style={draw, circle, inner sep=2pt, minimum size=13pt}, initial text=,initial where=above]

    \node[st] (s0) {$s_0$};
	\node[st, initial, right=2cm of s0] (s1) {$s_1$};
	\node[st, right=2cm of s1] (s2) {$s_2$};

	\draw[->] (s1) edge node[above] {$2-x$} (s0);
	\draw[->] (s1) edge node[above] {$-x$} (s2);
	 \draw[->] (s0) edge [loop left] node[left] {$1$} (s0);
	 \draw[->] (s2) edge [loop right] node[right] {$1$} (s2);
\end{tikzpicture}

%% file: figures/parametric_mini_pmc.tex
\begin{tikzpicture}[every node/.style={scale=1, font=\scriptsize}, st/.style={draw, circle, inner sep=2pt, minimum size=13pt}, baseline=(s0)]
		\node[st , initial, initial text=] (s0) {$s_0$} ;
		\node[st , right=0.8cm of s0] (s1) {$s_1$} ;
		\node[st , above=0.6cm of s1, accepting] (s2) {$s_2$} ;
		\node[st , right=0.8cm of s1] (s3) {$s_3$} ;
		
		\draw[->] (s0) -- node[below] {$x$} (s1);
		\draw[->] (s1) -- node[right] {$y$} (s2);
		\draw[->] (s0) -- node[left] {$1-x$} (s2);
		\draw[->] (s1) -- node[below] {$1-y$} (s3);
		\draw[->] (s2) edge[loop right] node {$1$} (s2);
		\draw[->] (s3) edge[loop above] node {$1$} (s3);
	\end{tikzpicture}

%% file: figures/parametric_mini_pmdp.tex
\begin{tikzpicture}[every node/.style={scale=1, font=\scriptsize}, st/.style={draw, circle, inner sep=2pt, minimum size=13pt},baseline=(s0)]
		\node[st, initial, initial text=] (s0) {$s_0$} ;
		\node[st,, right=0.8cm of s0] (s1) {$s_1$} ;
		\node[st, above=0.6cm of s1, accepting] (s2) {$s_2$} ;
		\node[st, right=0.8cm of s1] (s3) {$s_3$} ;
		
		\node[draw, circle,  inner sep=0.5pt, fill, right=0.3cm of s0] (s0alpha)  {};
		\node[draw, circle,  inner sep=0.5pt, fill, above=0.81cm of s0] (s0beta) {} ;
		
		\draw[-] (s0) -- node[below] {$\alpha$} (s0alpha);
		\draw[-] (s0) -- node[left] {$\beta$} (s0beta);
		\draw[->] (s0beta) -- node[above] {$1-y$} (s2);
		\draw[->] (s0beta) edge[bend left=110] node[above] {$y$} (s3);
		
		\draw[->] (s0alpha) -- node[below] {$x$} (s1);
		\draw[->] (s1) -- node[right] {$y$} (s2);
		\draw[->] (s0alpha) -- node[left] {$1-x$} (s2);
		\draw[->] (s1) -- node[below] {$1-y$} (s3);
		
		\draw[->] (s2) edge[loop right] node {$1$} (s2);
		\draw[->] (s3) edge[loop above] node {$1$} (s3);
		
\end{tikzpicture}

%% file: 4_qualitative.tex
Table~\ref{tab:qualcomplexity} summarises the results for qualitative reachability in (simple) pMDPs and pMCs.
Consider a pMDP $\pmdp$ and let $R \subseteq \wdval{\pmdp}$ and $\quant \in
\{\exists,\forall\}$. For convenience, we give names to the questions asking
whether there exists some pMDP $\mdp \in \generator[R]{\pmdp}$ with the
following properties. 
\begin{itemize}
	\item \emph{Positive reachability:}
		\( 
		\quant\;\sched \in \DMSched{\pmdp}:\;
		\pr{\instantiated{\pmdp}{\val}}{\sched}{\eventually\target} > 0
		\)
	\item \emph{Unsure reachability:}
	\( 
	\quant\;\sched \in \DMSched{\pmdp}:\;
	\pr{\instantiated{\pmdp}{\val}}{\sched}{\eventually\target} < 1
	\)
  	\item \emph{Safety:}
	    \( 
	      \quant\;\sched \in \DMSched{\pmdp}:\;
	      \pr{\instantiated{\pmdp}{\val}}{\sched}{\eventually\target} \leq 0
	    \)
  \item \emph{Almost-sure reachability:}
    \( 
      \quant\;\sched \in \DMSched{\pmdp}:\;
      \pr{\instantiated{\pmdp}{\val}}{\sched}{\eventually\target} \geq 1.
    \)
  
\end{itemize}
Note that the decision problem changes depending on $R$ and, for
pMDPs, on $\quant$. Together, these problems are the \emph{qualitative
reachability problems}.
Table \ref{tab:qualcomplexity} lists their computational complexity for a fixed number of
parameters, and the complexity if the parametric model contains arbitrarily
many parameters. In the latter case, we make a distinction based on whether
the parameter valuations range over the well-defined, graph-preserving or Boolean
valuations.

\begin{table}[tb]
\centering
\setlength{\tabcolsep}{5pt}
\begin{tabular}{|l|c|c|c|c|}
\hline
                               & Fixed \#  &
\multicolumn{3}{c|}{Arbitrary \# parameters} \\
\cline{3-5}
 & parameters & graph-preserving & well-defined & Boolean \\
\hline
\hline
  $> 0$ &  in \Ptime{} \refsize{Thm~\ref{thm:easy-cases}}
        & in \Ptime{} \refsize{Thm~\ref{thm:easy-cases}} 
        & in \Ptime{} \refsize{Thm~\ref{thm:easy-cases}}
        & \makecell{\NP-complete\\
          \refsize{Thm~\ref{thm:gen-upper-bnd}, Prop~\ref{pro:aspects:posreachabilityBoolean}}} \\
  $< 1$ & ''
        & ''
        & \makecell{\NP-complete\\
          \refsize{Thm~\ref{thm:gen-upper-bnd}, Prop~\ref{pro:aspects:unsurereachability}}}
        & '' \\
  $\leq 0$ & ''
           & ''
           & \makecell{\NP-complete\\
             \refsize{Thm~\ref{thm:gen-upper-bnd}, Prop~\ref{pro:complexity:qualitativereach}}}
           & '' \\
  $\geq 1$ & ''
           & ''
           & ''
           & '' \\
\hline
\end{tabular}
\caption{The complexity landscape for qualitative reachability in
  simple pMCs and pMDPs. Observe that the decision problems for pMCs and
  pMDPs (for maximal and minimal reachability probability values) are different, but
  (with respect to the considered classes) the categorisation coincides.
  Unlisted combinations of comparison operators and thresholds yield trivial
  decision problems.
}

\label{tab:qualcomplexity}
\end{table}

\subsection{Upper bounds}
Towards a general upper bound, recall that inside graph-consistent valuation
sets the
sets of vanishing transitions are invariant. The following is a corollary of
Proposition~\ref{pro:encoding:qualitativegraphproperty} and the fact that
(maximal and minimal) reachability values in MDPs are computable in polynomial time.
\begin{theorem}\label{thm:gen-upper-bnd}
  The qualitative reachability problems for simple pMDPs are all decidable in \NP.
\end{theorem}
Indeed, one can guess a graph-consistent set of valuations by, for instance, guessing an
assignment of the parameters with values $0$, $\nicefrac{1}{2}$, or $1$, for all of
them. In the instantiated MDP one can verify that the property holds in
polynomial time.

There are three particular cases in which the problem is tractable: when
considering graph-preserving valuations only, when the problem is positive
reachability, and when the number of parameters is fixed.
\begin{theorem}\label{thm:easy-cases}
  The following problems for simple pMDPs are decidable in polynomial time:
  \begin{itemize}
    \item all the qualitative reachability problems with respect
      to graph-preserving valuations;
    \item the positive reachability problems that include graph-preserving valuations; and  
    \item all the qualitative reachability problems for a fixed
      number of parameters.
  \end{itemize}
\end{theorem}
The main idea behind the proof is the same as for the previous claim.
Indeed, one can guess a graph-consistent set of valuations by choosing a `dummy variable
assignment' giving a value of $0$, $\nicefrac{1}{2}$, or $1$, for all of them.
We observe that the set of all such valuations forms a finite partition of the set of
well-defined valuations:
\begin{lemma}\label{lem:decomposingwdtogp}
  Let $\pmdp$ be a pMDP with parameters $X$.
  The set $\wdval{\pmdp}$ may be partitioned into at most $3^{|X|}$
  maximal graph-consistent sets of valuations.
\end{lemma}

\noindent We can now argue that the theorem holds.
\begin{proof}[Proof of Theorem~\ref{thm:easy-cases}]
  If we only consider graph-preserving valuations then the structure remains
  fixed. Hence, the qualitative reachability problems are essentially equivalent to their
  parameter-free counterparts (obtained, e.g., by assigning $\frac{1}{2}$ to
  all parameters) and therefore in \Ptime.

  For positive reachability, we observe that removing transitions is never
  beneficial and (non-empty) graph-preserving valuations are, in that sense, optimal
  for positive reachability --- see Lemma~\ref{lem:more-better}. Hence, the
  positive reachability problems can be decided by considering any
  graph-preserving instantiation (e.g., assigning $\frac{1}{2}$ for
  all parameters). Therefore, for well-defined instantiations, positive
  reachability is in \Ptime.

  When we have a fixed number of parameters, even when ranging over all
  well-defined instantiations, there are only constantly many different
  graph-consistent valuation sets --- see Lemma~\ref{lem:decomposingwdtogp}.
  Consequently, the problem reduces to a constant number of
  problems in \Ptime.
\end{proof}

\noindent
In the sequel we give \NP-lower bounds for the remaining cases.

\subsection{Lower bounds for Boolean valuations}

This hardness result crucially depends on the absence of graph-preserving
instantiations and is inspired by a construction in
\cite{baiercomplexity}.
\begin{proposition} \label{pro:aspects:posreachabilityBoolean}
  The qualitative reachability problems with respect to Boolean valuations are
  \NP-hard even for acyclic simple pMCs.
\end{proposition}

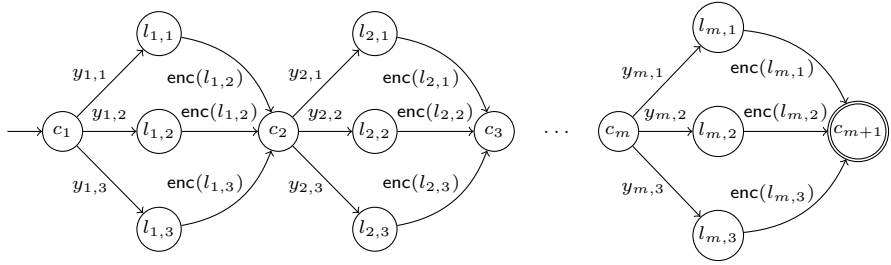
\begin{figure}
  \centering
  \input{figures/aspects_reachability_nphard}
  \caption{pMC construction for \NP-hardness of positive reachability in pMCs.}	
  \label{fig:aspects:positivereachabilitynphard}
\end{figure}

\begin{proof}
  We show a reduction from 3SAT to prove positive reachability is \NP-hard and
  comment on how to adapt the argument for the other problems.
	Let
  \[ \psi \coloneqq c_1 \wedge \dots \wedge c_m \]
  be a given 3SAT-formula, i.e. the clauses $c_i$ are of the form
  \[ c_i = l_{i,1} \vee l_{i,2} \vee l_{i,3}, \]
  where the $l_{i,j}$ are \emph{literals} (variables or negated variables).
  Let $\vars = \{ x_1, \dots, x_k \}$ be the variables of $\psi$.
  The pMC for $\psi$ is outlined in Figure~\ref{fig:aspects:positivereachabilitynphard}. 
  Formally, the pMC $\pmc \coloneqq (S, \init, \pars, \pprob)$ is defined as
  follows: The $4m + 2$ states
  \[
    S \coloneqq \{ c_i \mid 1 \leq i \leq m+1 \} \cup \{ l_{i,j} \mid 1 \leq i
    \leq m, 1 \leq j \leq 3 \} \cup \{ \bot \} \text{ with } \init \coloneqq c_1,
  \]
  $3m + k$ parameters
  \[
    \pars \coloneqq \{ \tilde{x} \mid x \in \vars \} \cup \{ y_{i,j} \mid  1
    \leq i \leq m, 1 \leq j \leq 3 \},  
  \]
  and with transitions
  \begin{align*}
    \pprob(s,s') \coloneqq \begin{cases}
      y_{i,j}     & \text{if }s = c_i, s' = l_{i,j}\text{ for some }1 \leq i \leq m, 1 \leq j \leq 3, \\
      \enc(l_{i,j}) & \text{if }s = l_{i,j}, s' = c_{i+1}\text{ for some }1 \leq i \leq m, 1 \leq j \leq 3, \\
      1-\enc(l_{i,j}) & \text{if }s = l_{i,j}, s' = \bot\text{ for some }1 \leq i \leq m, 1 \leq j \leq 3, \\
      0 & \text{otherwise,} 
    \end{cases}
  \end{align*}
  using
  \[
    \enc(l_{i,j}) \coloneqq \begin{cases}
    \tilde{x} & \text{if }l_{i,j} = x, \\
    1-\tilde{x} & \text{if } l_{i,j} = \overline{x}. 
    \end{cases}
   \]
   The target states are $\target = \{ c_{m+1} \}$. It should be clear that this
   construction can be realised in polynomial time.

  We will argue that $\psi$ is satisfiable if and only if there exists $\dtmc
  \in \generator[\pureval{\pmc}]{\pmc}$ such that
  $\pr{\dtmc}{}{\eventually\target} > 0$.
  Intuitively, the variables $l_{i,j}$ represent the witness literal for each satisfied clause, i.e., the literal that makes the clause true. 
  The parameters $\tilde{x}$ correspond to the $x$ variables in the 3SAT-formula as follows:
  For a valuation $\val$ of variables $\vars$ in $\psi$ and a valuation $\val'$ of $\pars$ such that $\val'(\tilde{x}) = 1 \text{ iff } \val(x) = \true$ it holds that:
  \[ \enc(l_{i,j})[\val']= 1 \iff \val(l_{i,j}) =\true. \]
  Formally, first assume there exists a satisfying assignment $\val$ for $\psi$.  Then, this assignment makes at least one
  literal $l_{i,*}$ in every clause $c_i$ true.  We consider $\val'$ with the
  corresponding $y_{i,*}$ assigned to $1$ and $\tilde{x}$ assigned $1$ iff
  $\val(x) = \true$.  Then, in the MC $\instantiated{\pmc}{\val'}$, there is a path
  from $\init$ to $\target$.
 
 Now assume that there exists an MC $\dtmc \in \generator[\mathbb{B}]{\pmc}$ with a path from $\init$ to $\target$.
 Observe that this path in $\dtmc$ is the only path to the target. 
 We construct a satisfying assignment $\val$ for $\psi$. 
 This path goes through a set of $l_{i,*}$.
 These become the witness literals that make all the clauses true. 
 The assignment to the variables $\vars$ are obtained from the occurrences of
 $\tilde{x}$ along the path, or equivalently, by lookup from the witness
 literals given by the path.
 
 For safety, almost-sure, and unsure reachability, we observe that the
 probability to reach $c_{m+1}$ in $\pmc$ is either zero or one for any
 Boolean valuation so the corresponding proofs are straightforward
 adaptions of the one given above.
\end{proof}

\subsection{Lower bounds for well-defined valuations}
We have argued that positive reachability is in \Ptime. We now show that all
other qualitative problems are \NP-complete. We begin with the almost-sure
reachability and safety problems.
\begin{proposition}[From~\cite{DBLP:conf/rp/Chonev19}]
  \label{pro:complexity:qualitativereach}
	The safety and almost-sure reachability problems are \NP-hard even for
  simple pMCs.
\end{proposition}
\begin{figure}
  \centering
  \input{figures/aspects_chonev_construction}	
  \caption{pMC construction for \NP-hardness of almost-sure reachability in pMCs.}
  \label{fig:aspects:chonevconstruction}
\end{figure}
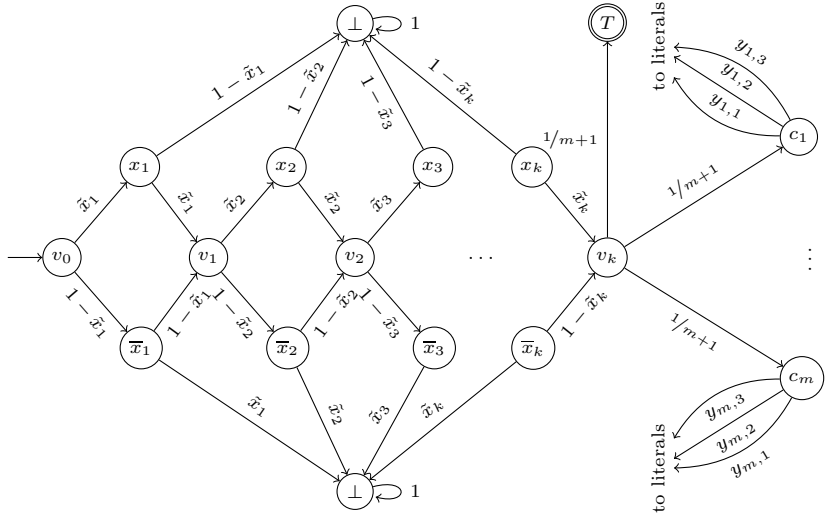

We deliberately recall the proof from~\cite{DBLP:conf/rp/Chonev19}
rather than adapting the construction used to prove
Proposition~\ref{pro:aspects:posreachabilityBoolean}, as the former is a
crucial step towards Proposition~\ref{pro:e_reach_gr_np_hard}.
The essential idea here is to enforce Boolean
valuations.

\begin{proof}
  We reduce from 3SAT once more. To that end, let \( \psi = c_1 \wedge \dots
  \wedge c_m \) be a given 3SAT-formula with clauses $c_i$ of the form
  \( c_i = l_{i,1} \vee l_{i,2} \vee l_{i,3} \) and variable set $\vars = \{ x_1, \dots,
  x_\nrvars \}$.
	The pMC for $\psi$ is
  outlined in Figure~\ref{fig:aspects:chonevconstruction}, where state $\bot$ is duplicated to avoid clutter.  Formally, the pMC
  $\pmc = (S, \init, \pars, \pprob)$ is defined as follows:
  \[
    S \coloneqq
    \{v_i \mid 0 \leq i \leq k\}~\uplus~
    \{x_i, \overline{x_i} \mid 1 \leq i \leq k\}~\uplus~
    \{c_i \mid 1 \leq i \leq m\}~\uplus~
    \{T, \bot\}
  \]
  are the $3k + m + 3$ states, $v_0 = \init$ is the initial state, $T$ and
  $\bot$ indicate target and sink respectively,
  \[
    \pars \coloneqq \{\tilde{x} \mid x \in \vars \}~\cup~\{y_{i,j} \mid 1 \leq
    i \leq m, 1 \leq j \leq 3 \}
  \]
  are the $k + 3m$ parameters, for all $1 \leq i \leq m$ and $1 \leq j \leq k$
  we define the transition probabilities as
  \begin{align*}
    &\pprob(v_{i-1}, x_i) \coloneqq \tilde{x}_i, && \pprob(v_{i-1}, \overline{x_i}) \coloneqq 1-\tilde{x}_i, \\
    &\pprob(x_i, v_i) \coloneqq \tilde{x}_i, && \pprob(\overline{x_i}, v_i) \coloneqq 1-\tilde{x}_i, \\
    &\pprob(x_i, \bot) \coloneqq 1-\tilde{x}_i, && \pprob(\overline{x_i}, \bot) \coloneqq \tilde{x}_i, \\
    &\pprob(v_k, c_i) \coloneqq \frac{1}{m+1}, && \pprob(v_k, T) \coloneqq
    \frac{1}{m+1}, \\
    &\pprob(c_j, x_i) \coloneqq y_{j,r} \text{ if } l_{j,r} = x_i \text{ (in }\psi\text{)}, && \pprob(c_j, \overline{x}_i) \coloneqq y_{j,r} \text{ if } l_{j,r} = \overline{x}_i \text{ (in }\psi\text{)}.
  \end{align*}
  We let $\pprob(s,t) = 0$ for each pair $(s,t)$ of states not specified above.

  Observe that under any well-defined valuation, there are exactly two bottom
  strongly connected components, namely $\bot$ and $\target$. As a
  consequence:
  \begin{equation} \label{eq:gadget:twoendcomponents}
    \text{for all } \dtmc \in \generator{\pmc}:\; \quad
    \pr{\dtmc}{}{\eventually\target} + \pr{\dtmc}{}{\eventually\{ \bot \}} =
    1.
  \end{equation}
  
  We will argue that $\psi$ is satisfiable if and only if there exists $\dtmc
  \in \generator{\pmc}$ such that $\pr{\dtmc}{}{\eventually\target} \geq 1$.
  For convenience, we write $1$ and $0$ instead of $\true$ and $\false$
  respectively.

  First, assume $\psi$ is satisfiable.  Choose some satisfying assignment
  $\val$ for $\psi$. We construct $\val' \in \wdval{\pmc}$ in two steps.
  First, let  $\val'(\tilde{x}_i) = \val(x_i) \in \{ 0, 1 \}$ for all $1 \leq
  i \leq k$. Thus, $\bot$ is unreachable. Second, for each clause $c_i$,
  select one literal $l_{i,j}$ which makes $c_i$ true, and set $\val'(y_{i,j})
  = 1$. Set all other $y_{i,j}$ to $0$. It follows from
  Equation~\eqref{eq:gadget:twoendcomponents} that
  $\pr{\pmc[\val']}{}{\eventually\target} = 1$.

  Now, assume there is a well-defined valuation $\val'$ such that
  $\pr{\instantiated{\pmc}{\val'}}{}{\eventually\target} \geq 1$. Then, using
  Equation~\eqref{eq:gadget:twoendcomponents}, no path leads to $\bot$.  For
  $\tilde{x}_i$, that means that necessarily $\val'(\tilde{x}_i) \in \{ 0, 1
  \}$. Note that $\val'$ must be such that we can choose for each $c_i$ a literal
  $l_{i,j}$ (a $y_{i,j}$ set to $1$) which surely reaches $v_k$ again.
  These $l_{i,j}$ are exactly the witness literals making every clause
  true. It follows that the assignment for $\tilde{x}$ gives rise to a satisfying
  valuation $\val$ for $\psi$. 
 
  To conclude we observe that the construction can be easily adapted to show
  \NP-hardness of the safety problem.
\end{proof}

To close this section we describe how to adapt the construction used to prove
Proposition~\ref{pro:aspects:posreachabilityBoolean} in order to show the
unsure reachability problem is also \NP-hard.
\begin{proposition} \label{pro:aspects:unsurereachability}
	The unsure reachability problems are \NP-hard even for
  simple pMCs.
\end{proposition}
\begin{proof}[Proof sketch]
  We reuse the pMC from Proposition~\ref{pro:aspects:posreachabilityBoolean} depicted in Figure~\ref{fig:aspects:positivereachabilitynphard}
  and extend it with a transition from $c_{m+1}$ to $c_1$ with probability
  $1$. Further, we set $\bot$ as the only target state.

  Observe that if $c_{m+1}$ is reached with probability $1$ then no
  probability `leaks' to the $\bot$-state. Hence the target is reached with
  probability $0$. Otherwise, the target is the only bottom strongly
  connected component and the probability to reach that becomes $1$. The
  result thus follows from an almost identical argument as the one given for
  Proposition~\ref{pro:aspects:posreachabilityBoolean}.
\end{proof}

%% file: figures/aspects_reachability_nphard.tex
\begin{tikzpicture}[scale=1, every node/.style={scale=1, font=\scriptsize},
  st/.style={draw, circle, inner sep=1pt, minimum size=15pt}, initial text=]
\node[st, initial] (c0) {$c_1$};
\node[st, right=0.7cm of c0] (l02) {$l_{1,2}$};
\node[st, above=0.7cm of l02] (l01) {$l_{1,1}$};
\node[st, below=0.7cm of l02] (l03) {$l_{1,3}$};

\node[st, right=1cm of l02] (c1) {$c_2$};
\node[st, right=0.7cm of c1] (l12) {$l_{2,2}$};
\node[st, above=0.7cm of l12] (l11) {$l_{2,1}$};
\node[st, below=0.7cm of l12] (l13) {$l_{2,3}$};

\node[st, right=1cm of l12] (c2) {$c_3$};
\node[right=0.25cm of c2] (c3) {$\hdots$};
\node[st, right=0.25cm of c3] (cm) {$c_m$};
\node[st, right=0.7cm of cm] (lm2) {$l_{m,2}$};
\node[st, above=0.7cm of lm2] (lm1) {$l_{m,1}$};
\node[st, below=0.7cm of lm2] (lm3) {$l_{m,3}$};
\node[st, right=1.1cm of lm2, accepting] (t) {$c_{m+1}$};

\draw[->] (c0) edge node[above left=-1.5mm] {$y_{1,1}$} (l01);
\draw[->] (c0) edge node[above] {$y_{1,2}$} (l02);
\draw[->] (c0) edge node[below left=-1.5mm] {$y_{1,3}$} (l03);

\draw[<-] (c1) edge[bend right] node[pos=0.3,left] {$\enc(l_{1,2})$} (l01);
\draw[<-] (c1) edge node[above] {$\enc(l_{1,2})$} (l02);
\draw[<-] (c1) edge[bend left]  node[pos=0.3,left] {$\enc(l_{1,3})$} (l03);

\draw[->] (c1) edge node[above left=-1.5mm] {$y_{2,1}$} (l11);
\draw[->] (c1) edge node[above] {$y_{2,2}$} (l12);
\draw[->] (c1) edge node[below left=-1.5mm] {$y_{2,3}$} (l13);

\draw[<-] (c2) edge[bend right] node[pos=0.3,left] {$\enc(l_{2,1})$} (l11);
\draw[<-] (c2) edge node[above] {$\enc(l_{2,2})$} (l12);
\draw[<-] (c2) edge[bend left]  node[pos=0.3,left] {$\enc(l_{2,3})$} (l13);

\draw[->] (cm) edge node[above left=-1.5mm] {$y_{m,1}$} (lm1);
\draw[->] (cm) edge node[above] {$y_{m,2}$} (lm2);
\draw[->] (cm) edge node[below left=-1.5mm] {$y_{m,3}$} (lm3);

\draw[<-] (t) edge[bend right] node[pos=0.3,left] {$\enc(l_{m,1})$} (lm1);
\draw[<-] (t) edge node[above] {$\enc(l_{m,2})$} (lm2);
\draw[<-] (t) edge[bend left] node[pos=0.3,left] {$\enc(l_{m,3})$} (lm3);

\end{tikzpicture}

%% file: figures/aspects_chonev_construction.tex
	\begin{tikzpicture}[scale=1, every node/.style={scale=1, font=\scriptsize}, st/.style={draw, circle, inner sep=2pt, minimum size=13pt}, initial text=]
		\node[st,initial] (v_0) {$v_0$};
		\node[st, right=5mm of v_0,yshift=12mm] (x_1) {$x_1$};
		\node[st, right=5mm of v_0,yshift=-12mm] (notx_1) {$\overline{x}_1$};
		
		\node[st, right=12mm of v_0,xshift=6pt] (v_1) {$v_1$};
		\node[st, right=5mm of v_1,yshift=12mm] (x_2) {$x_2$};
		\node[st, right=5mm of v_1,yshift=-12mm] (notx_2) {$\overline{x}_2$};
		
		\node[st, right=12mm of v_1,xshift=6pt] (v_2) {$v_2$};
		\node[st, right=5mm of v_2,yshift=12mm] (x_3) {$x_3$};
		\node[st, right=5mm of v_2,yshift=-12mm] (notx_3) {$\overline{x}_3$};
		
		\node[right=11mm of v_2] (dots) {$\ldots$};
		\node[st, right=1mm of dots,yshift=12mm] (x_m) {$x_k$};
		\node[st, right=1mm of dots,yshift=-12mm] (notx_m) {$\overline{x}_k$};
		
		\node[st, right=11mm of dots] (v_m) {$v_k$};
		
		\node[st, right=of v_m, xshift=10mm,yshift=16mm] (phi_1) {$c_1$};
		\node[st, right=of v_m, xshift=10mm,yshift=-16mm] (phi_k) {$c_m$};
		\node[right=of v_m, xshift=14mm,rotate=90,anchor=center] (dots2) {$\dots$};
		
		\node[left=of phi_1, xshift=-6mm,yshift=12mm,rotate=90,anchor=center] (literals_1) {to literals};
		\node[left=of phi_k, xshift=-6mm,yshift=-12mm,rotate=90,anchor=center] (literals_k) {to literals};
		
		\node[st,above=of v_2,yshift=16mm] (bot1) {$\bot$};
		\node[st,below=of v_2,yshift=-16mm] (bot2) {$\bot$};
		
		\node[st,accepting,above=of v_m,yshift=16mm] (T) {$T$};
		
		%%%%%%%%%%%%%%%%
		
		\draw[->] (v_0) -- node[sloped, anchor=center,above] {$\tilde{x}_1$} (x_1);
		\draw[->] (v_0) -- node[sloped, anchor=center,below] {$1-\tilde{x}_1$} (notx_1);
		\draw[->] (x_1) -- node[sloped, anchor=center,above] {$\tilde{x}_1$} (v_1);
		\draw[->] (notx_1) -- node[sloped, anchor=center,below] {$1-\tilde{x}_1$} (v_1);
		
		\draw[->] (v_1) -- node[sloped, anchor=center,above] {$\tilde{x}_2$} (x_2);
		\draw[->] (v_1) -- node[sloped, anchor=center,below] {$1-\tilde{x}_2$} (notx_2);
		\draw[->] (x_2) -- node[sloped, anchor=center,above] {$\tilde{x}_2$} (v_2);
		\draw[->] (notx_2) -- node[sloped, anchor=center,below] {$1-\tilde{x}_2$} (v_2);
		
		\draw[->] (v_2) -- node[sloped, anchor=center,above] {$\tilde{x}_3$} (x_3);
		\draw[->] (v_2) -- node[sloped, anchor=center,below] {$1-\tilde{x}_3$} (notx_3);
		
		\draw[->] (x_m) -- node[sloped, anchor=center,above] {$\tilde{x}_k$} (v_m);
		\draw[->] (notx_m) -- node[sloped, anchor=center,below] {$1-\tilde{x}_k$} (v_m);
		
		\draw[->] (v_m) -- node[sloped, anchor=center,above] {$\nicefrac{1}{m+1}$} (phi_1);
		\draw[->] (v_m) -- node[sloped, anchor=center,below] {$\nicefrac{1}{m+1}$} (phi_k);
		
		\draw[->] (phi_1) edge[bend left=30] node[sloped, anchor=center,above] {$y_{1,1}$} (literals_1);
		\draw[->] (phi_1) edge[bend right=30] node[sloped, anchor=center,above] {$y_{1,3}$} (literals_1);
		\draw[->] (phi_1) edge[bend left=0] node[sloped, anchor=center,above] {$y_{1,2}$} (literals_1);
		
		\draw[->] (phi_k) edge[bend left=30] node[sloped, anchor=center,below] {$y_{m,1}$} (literals_k);
		\draw[->] (phi_k) edge[bend right=30] node[sloped, anchor=center,below] {$y_{m,3}$} (literals_k);
		\draw[->] (phi_k) edge[bend left=0] node[sloped, anchor=center,below] {$y_{m,2}$} (literals_k);
		
		\draw[->] (x_1) -- node[sloped, anchor=center,above] {$1-\tilde{x}_1$} (bot1);
		\draw[->] (x_2) -- node[sloped, anchor=center,above] {$1-\tilde{x}_2$} (bot1);
		\draw[->] (x_3) -- node[sloped, anchor=center,below] {$1-\tilde{x}_3$} (bot1);
		\draw[->] (x_m) -- node[sloped, anchor=center,above] {$1-\tilde{x}_k$} (bot1);
		
		\draw[->] (notx_1) -- node[sloped, anchor=center,above] {$\tilde{x}_1$} (bot2);
		\draw[->] (notx_2) -- node[sloped, anchor=center,above] {$\tilde{x}_2$} (bot2);
		\draw[->] (notx_3) -- node[sloped, anchor=center,above] {$\tilde{x}_3$} (bot2);
		\draw[->] (notx_m) -- node[sloped, anchor=center,above] {$\tilde{x}_k$} (bot2);
		
		\draw[->] (v_m) -- node[left] {$\nicefrac{1}{m+1}$} (T);
		
		\draw[->] (bot1) edge [loop right] node[right] {$1$} (bot1);
		
		\draw[->] (bot2) edge [loop right] node[right] {$1$} (bot2);
	\end{tikzpicture}	
	

%% file: 5_quantitative.tex
\begin{table}[tb]
\centering
\setlength{\tabcolsep}{3pt}
\begin{tabular}{|ll|c|c|c|}
\hline
           &                      & Fixed \#  &
\multicolumn{2}{c|}{Arbitrary \# parameters} \\
\cline{4-5}
& & parameters & well-defined & graph-preserving \\
\hline
\hline
\parbox[t]{3.5mm}{\multirow{5}{*}{\rotatebox[origin=c]{90}{pMC}}} &
  $\reachdp{\geq/\leq}$ & \makecell{in
  \Ptime\\\refsize{Thm~\ref{thm:aspects:ptimepmcs}}} &
  \multicolumn{2}{c|}{--- \ETR-complete
  \refsize{Thm~\ref{thm:etr:pmcsarehard}} ---} \\
& $\reachdp{>}$ & '' &
  \makecell{\NP-hard\\\refsize{Prop~\ref{pro:e_reach_gr_np_hard}}} &
  \makecell{$\reachdp{>}_\mathrm{wd}$-complete\\\refsize{Prop~\ref{pro:e_reach_less_red_e_reach_less_gp},
  Prop~\ref{pro:e_reach_gp_red_e_reach}}} \\
& $\reachdp{<}$ & '' &
  \makecell{\NP-hard\\\refsize{Prop~\ref{pro:e_reach_gr_np_hard}}}
&\makecell{$\reachdp{>}_\mathrm{wd}$-complete\\\refsize{Prop~\ref{pro:complexity:restrictrelations}}}
  \\
\hline
\parbox[t]{3.5mm}{\multirow{6}{*}{\rotatebox[origin=c]{90}{pMDP}}} &
  $\exists\reachdp{\geq/\leq}$ & \makecell{in \NP\\
  \refsize{Prop~\ref{pro:fp_ee_reach_in_np}}} & \multicolumn{2}{c|}{---
  \ETR-complete~\refsize{(trivial)} ---}  \\
& $\exists\reachdp{>}$ & '' & \multicolumn{2}{c|}{---
  $\reachdp{>}_\mathrm{wd}$-complete
  \refsize{Prop~\ref{pro:existwdgeeqexistwdge}, Prop~\ref{pro:eeggpeqerggwd}} ---}
  \\
& $\exists\reachdp{<}$  & '' &
  \makecell{$\reachdp{<}_\mathrm{wd}$-complete\\\refsize{Prop~\ref{pro:existwdgeeqexistwdge}}}
  &  \makecell{ $\reachdp{>}_\mathrm{wd}$-hard \\\refsize{(trivial)}}\\
& $\forall\reachdp{\bowtie}$ & \makecell{in
  \NP~\\\refsize{Thm~\ref{thm:fp_ea_reach_in_np}}} & \multicolumn{2}{c|}{---
  \ETR-complete \refsize{Thm~\ref{thm:etr:mdpsarehard}} ---}\\
\hline 
\end{tabular}
\caption{The complexity landscape for \emph{quantitative} reachability in
simple pMDPs. All problems are in \ETR.}

\label{tab:complexity}
\end{table}

Table~\ref{tab:complexity} summarises the results we will present in this
section. 
We use the following notation for conciseness:
For $\quant \in \{\exists,\forall\}$ and 
$\mathrm{\bowtie} \in
\{\leq, <, >, \geq\}$, let
\[ 
  \quant\reachdp{\bowtie}_{\mathrm{wd}} \stackrel{\mathrm{def}}{\iff}
  \exists\;\val \in \wdval{\pmdp},
  \quant\;\sched \in \DMSched{\pmdp}:\;
    \pr{\instantiated{\pmdp}{\val}}{\sched}{\eventually\target} \bowtie \frac{1}{2}
\]
be the \emph{quantitative reachability problems}. We write
$\quant\reachdp{\bowtie}_{\mathrm{gp}}$ whenever we consider graph-preserving
instantiations. We write $\quant\reachdp{\bowtie}_{*}$ to denote both
the $\mathrm{wd}$ and $\mathrm{gp}$ variants. Furthermore, if $\pmdp$ is a pMC
we omit the quantifier, e.g. $\reachdp{<}_{*}$.

\paragraph{Fixed threshold}
\label{page:fixedthreshold}
Note that we have fixed a threshold of
$\nicefrac{1}{2}$. This is without loss of generality as any \emph{given}
rational threshold $\lambda$ may be reduced to $\nicefrac{1}{2}$: Simply
prepend a transition with probability $\nicefrac{1}{2}$ to the initial state,
one with probability $\nicefrac{1}{2}(1-\lambda)$ to the target state and a
third one with probability $\nicefrac{1}{2}\lambda$ to a sink state. Then it
can be readily checked that the reachability probability in the original model compares to $\lambda$ in some desired way iff it compares to $\nicefrac{1}{2}$ in the modified model.

We first show that well-defined and graph-preserving sets of valuations are
\emph{semialgebraic}, i.e., they can be described by an ETR formula. Then we give a detailed account on how to encode the
reachability problems for pMCs into the ETR.  First, we consider reachability
probabilities and the easier case of graph-preserving valuation subsets,
then in general for well-defined valuation subsets.
We then show the lifted encodings to pMDPs.

\subsection{ETR encoding for pMCs}
Below, we show that sets of all well-defined or graph-preserving valuations
are indeed semialgebraic.
The following set of constraints is a natural encoding of Definition~\ref{def:parametric:welldefined}.
\begin{constraints}[Well-defined sets of valuations]
\label{constraints:aspects:encoding:polynomialwelldefined}
The following constraints capture well-defined valuations for a polynomial pMDP $\pmdp$:
\begin{alignat*}{3}
   & 0 \cleq \pprob(s,\act,s')  \cleq 1 & & \text{ for all }s,s' \in S, \act \in \EnAct(s)\quad\text{(with }\pprob(s,\act,s') \neq 0\text{)}, \\
   & \sum_{s'\in S} \pprob(s,\act,s') = 1 & & \text{ for all }s \in S,  \act \in \EnAct(s).
\end{alignat*}	
We denote the corresponding formula for this constraint system with $\regenc{\pmdp}{\mathrm{wd}}$.
\end{constraints}
The constraints ensure that (1)~all (non-zero) transitions are evaluated to a
probability, and (2)~transition probabilities describe  distributions. It follows that the set of well-defined
valuations of some $\pmdp$ is semialgebraic.
\begin{example}
\label{ex:parametric:wdregion}
  Recall the pMC $\pmc$ for the Knuth-Yao die from
  Figure~\ref{fig:parametric:pkydiepmc}, with the well-defined valuations as
  in Example~\ref{ex:parametric:welldefined}.
  We have:
  \begin{alignat*}{2} \regenc{\pmc}{\mathrm{wd}} &=~& & x \cgeq 0 \,\land\, 1-x \cgeq 0  \, \land \, x+1-x \ceq 1 \\
      & \ \land~& & y \cgeq 0 \, \land \, 1-y \cgeq 0 \, \land \, y+1-y \ceq 1. 
  \end{alignat*}
  This formula simplifies to $0 \cleq x \cleq 1 \land 0 \cleq y \cleq 1$.

  Now recall the pMDP $\pmdp$ for rock-paper-scissors from
  Figure~\ref{fig:parametric:rpspmdp}, with the well-defined valuations as in
  Example~\ref{ex:parametric:welldefined}.  We get:
  \begin{alignat*}{2}
    \regenc{\pmdp}{\mathrm{wd}} &=~& &  x_R \cgeq 0 \,\land\, x_P \cgeq 0  \, \land \, x_S \cgeq 0 \, \land \, x_R+x_P+x_S \ceq 1 \\
      & \ \land~& &  x'_R \cgeq 0 \,\land\, x'_P \cgeq 0  \, \land \, x'_S \cgeq 0 \, \land \, x'_R+x'_P+x'_S \ceq 1.
  \end{alignat*}
\end{example}
This encoding is easily extended with strict inequalities to describe graph-preserving valuations, based on Definition~\ref{def:parametric:graph_preservation}.

We now move to the more interesting question of how to actually encode
reachability.  We start with pMCs, which we consider extensively as the ideas
for pMDPs are mostly straightforward extensions. 

\subsubsection{Qualitative analysis}
Before we treat quantitative problems, we start with the qualitative ones.
\begin{definition}
  Let $\pmc$ be a pMC. The \emph{zero-states} for valuation $\val \in
  \wdval{\pmc}$ are \[ \pzerostates[\target]{\val} \coloneqq \{ s \in S
  \mid \prsol[s]{\target}{\pmc}[\val] = 0 \} \] containing the states that
  reach the target with probability zero in instantiation $\val$ and the
  \emph{one-states} for valuation $\val \in \wdval{\pmc}$ is the set \[
  \ponestates[\target]{\val} \coloneqq \{ s \in S \mid
  \prsol[s]{\target}{\pmc}[\val] = 1 \} \] containing all states that reach the
  target almost surely.
\end{definition}
These sets vary for different valuations. However,
for any $\val, \val' \in \wdval{\pmc}$, \[ \vanish[\pmc]{\val} \subseteq
\vanish[\pmc]{\val'}\text{ implies }\pzerostates[\target]{\val}
\subseteq \pzerostates[\target]{\val'},\] and \[ \vanish[\pmc]{\val} =
\vanish[\pmc]{\val'}\text{ implies }\pzerostates[\target]{\val} =
\pzerostates[\target]{\val'}.\]
Essentially, removing transitions may cut states from having a path to the
target states, but never adds new paths. 

\paragraph{Computing the sets}
Proposition~\ref{pro:encoding:qualitativegraphproperty} justifies the notation
$\pzerostates[\target]{\region}$ and $\ponestates[\target]{\region}$ for graph-consistent $\region$ as being
the (unique) sets $\pzerostates[\target]{\val}$, $\ponestates[\target]{\val}$
for any $\val \in \region$, respectively. Crucially, for any fixed
graph-consistent valuation set, the sets $\pzerostates[\target]{\val}$ and
$\ponestates[\target]{\val}$ may be computed on the parameter-free
$\instantiated{\pmc}{\val}$.  However, when regarding a non-graph-consistent
valuation set, this does not necessarily suffice. 
The essential idea is to encode the graph-algorithm together
with a ranking function.

\subsubsection{Quantitative analysis}
We move from the qualitative setting to a quantitative one. The principle is
again to generalise the parameter-free case.  We develop the
encoding in two steps: First, we consider an encoding only valid for
graph-preserving valuations. In particular, it requires the zero-states to be
known a-priori. Later, we combine this encoding with the earlier qualitative
encodings to compute the zero-states on the fly.

\paragraph{Graph-preserving case}
We lift the classical equation system for parameter-free MCs to polynomial pMCs. 

\begin{constraints}
  Let $\pmc$ be a polynomial pMC. 
  We assume a graph-preserving valuation set $\region \subseteq \gpval{\pmc}$.
  Consider real variables $\{ p_s \mid s \in S \}$ and variables for the parameters $\pars$ of $\pmc$:
  \label{const:pmcequations}
  \begin{align*}
  & p_s \ceq 1 & &\text{for all } s \in \target, \\
  & p_s \ceq 0 & &\text{for all } s \in \pzerostates[\target]{\region}, \\
  & p_s \ceq \sum_{s' \in S} \pprob(s,s') \cdot p_{s'} & &\text{for all } s \in S \setminus \big(\target \cup \pzerostates[\target]{\region}).
\end{align*}
We denote the corresponding formula with $\regenc{\pmc}{\mathrm{gp}}$.
\end{constraints}
Note that the constraints do not actually depend on $\region$, only the fact that $\region$ is graph-preserving matters. The constraints are essentially identical to those for parameter-free MCs. 
The key difference is that the transition probabilities are no longer constants.
Therefore (in general\footnote{The notable exceptions are systems where the parameters only occur in states where \emph{all} successor states are sink- or target-states.}) the encodings are non-linear.

Recall that we have to restrict the parameter valuations accordingly and
encode that the induced probability in the initial state compares $\bowtie
\nicefrac{1}{2}$.  We add these constraints and obtain:
\begin{theorem} \label{thm:encoding:pmcgp}
  Let $\pmc$ be a polynomial pMC with target states $\target$ and let $\region
  \subseteq \gpval{\pmc}$ be a semialgebraic set given by $\psi_R$. We define
	\[ \psi \coloneqq \regenc{\pmc}{\mathrm{gp}} \land p_\init \bowtie \nicefrac{1}{2} \land \psi_R. \] 
  Then, for all $\val \in \Val$,
  \[
    \val \text{ satisfies } \psi \text{ iff }
    \val \in R \land
    \pr{\instantiated{\pmc}{\proj{\val}{\pars}}}{}{\eventually\target} \bowtie
    \nicefrac{1}{2}.
  \]
\end{theorem}

\paragraph{Well-defined case}
We extend the encoding to any well-defined valuation set.  An
essential assumption before was that the set of zero-states is fixed and may
be precomputed.  This assumption is no longer valid.  We thus encode the
computation of the zero-states using the encoding for positive reachability.
\begin{constraints}
 Let $\pmc$ be a polynomial pMC with states $S$.
Consider Boolean variables $\{ q_s \mid s \in S \}$, real variables $\{ p_s, r_s \mid s \in S \}$, and variables for the parameters:
	\begin{align*}
	& p_s = 1 & & \text{for all } s \in \target, \\
	& q_s \text{ is true} & &\text{for all } s \in \target, \\
	& q_s \leftrightarrow \bigvee_{s'\in S} \left( \pprob(s,s') \cgt 0 \land \left( q_{s'} \land r_{s} \clt r_{s'} \right) \right)	 \quad & &\text{for all } s \in S \setminus \target, \\
	& \neg q_s \rightarrow p_s = 0 & & \text{for all } s \in S \setminus \target, \\
	& q_s \rightarrow p_s \ceq \sum_{s' \in S} \pprob(s,s') \cdot p_{s'} & &\text{for all } s \in S \setminus \target .
\end{align*}
We denote the corresponding formula with $\regenc{\pmc}{\mathrm{wd}}$.
\end{constraints} The meaning of the variables is as before: The variables
$q_s$ determine whether we have to compute the non-zero probability to the
target or whether this probability is zero.  The $r_s$ variables are auxiliary
variables ranking the states.  The specialised constraints for the
graph-preserving case (Constraints~\ref{const:pmcequations}) are obtained by
setting all variables of non-zero states $q_s$ to $\true$.  The following
theorem is the analogue to Theorem~\ref{thm:encoding:pmcgp}.

\begin{theorem} \label{thm:encoding:pmcwd}
  Let $\pmc$ be a polynomial pMC with target states $\target$ and let $\region
  \subseteq \wdval{\pmc}$ be a semialgebraic set given by $\psi_R$. 
  We define
	\[ \psi \coloneqq \regenc{\pmc}{\mathrm{wd}} \land p_\init \bowtie \nicefrac{1}{2} \land \psi_R. \] 
  Then, for any $\val \in \Val$,
	\[
    \val \text{ satisfies } \psi \text{ iff }
    \val \in R \land
    \pr{\instantiated{\pmc}{\proj{\val}{\pars}}}{}{\eventually\target} \bowtie
    \nicefrac{1}{2}.
  \]
\end{theorem}

\subsubsection{Alternative encoding via solution functions}
The encodings presented above contain $\mathcal{O}(|S|+|\pars|)$ many
variables. As solving ETR is exponential in the number of variables, the large
number of variables is a significant hurdle.  In this section, we present
encodings that prevent the dependency on the number of states, by
incorporating the solution function. 

We reconsider the encoding for pMCs under the assumption that we may
precompute the zero-states.  Reachability in an MC corresponds to a linear
equation system $A \cdot \vec{p} = \vec{b}$ where $\vec{p}$ is the solution
vector, and $A$ is a matrix and $b$ is a vector (see, e.g.,~\cite{BK08}). 
For pMCs, $A \cdot \vec{p} = \vec{b}$ may be viewed as a linear equation
system over the field $\QQ(\pars)$ of rational functions with rational
coefficients.  That is, the entries of $A$ are no longer rational numbers, but
rational functions instead~\cite{baiercomplexity, DBLP:conf/ictac/Daws04}.

By basic linear algebra, we obtain that
for all pMCs $\pmc$ with targets $\target$, there exists $f \in \QQ(\pars)$
such that 
\[ \prsol{\target}{\pmc}[\val] = f[\val]\text{ for all }\val \in \gpval{\pmc}.\]
The rational function $f$ is exactly the entry $q_\init$ of the unique
solution $q$ for the system $A \cdot \vec{p} = \vec{b}$.  Thus, solving linear
equation systems (symbolically) is sufficient to find these solution
functions. 
We observe that $f$ is the restriction of $\prsol{\target}{\pmc}$ to
$\gpval{\pmc}$. We denote this restriction with $_\mathrm{gp}\prsol{\target}{\pmc}$.

We conclude this alternative ETR encoding by stating its main property.
\begin{theorem}\label{thm:encoding:solfuncencoding}
  Let $\pmc$ be a polynomial pMC with target states $\target$. Let $_\mathrm{gp}\prsol{\target}{\pmc} = \nicefrac{f}{g}$ for polynomials $f$ and $g$ be the solution function  of $\pmc$ on graph-preserving valuations and let $\region
  \subseteq \gpval{\pmc}$ be a semialgebraic set given by $\psi_R$. 
  We define
	\[
    \psi \coloneqq  \left( (g \cgt 0 \land f \bowtie \nicefrac{1}{2} \cdot g) \lor (g < 0
    \land \nicefrac{1}{2} \cdot g \bowtie f) \right) \land \psi_R.
  \] 
  Then, for any $\val \in \Val$,
	\[
    \val \text{ satisfies } \psi \text{ iff }
    \val \in R \land
    \pr{\instantiated{\pmc}{\proj{\val}{\pars}}}{}{\eventually\target} \bowtie
    \nicefrac{1}{2}.
  \]
\end{theorem}

\subsection{ETR encoding for pMDPs}
In this section, we generalise the encodings from pMCs to pMDPs.  We
distinguish between existential and universal nondeterminism.  Together, this
subsection establishes that for every pMDP the set of all valuations giving a
positive answer to reachability problems is semialgebraic.

\subsubsection{Qualitative analysis}
Again, we first give some preliminary considerations regarding the qualitative
case before moving to the quantitative setting.
\begin{definition}
  Let $\pmdp$ be a pMDP. The \emph{exist-zero states} for a valuation $\val
  \in \wdval{\pmdp}$ is the set
  \[ \pzeroexstates[\target]{\val} \coloneqq \{ s \in S \mid \exists \sched \in \DMSched{\pmdp}
  \text{ s.t. } \prsol[s]{\target}{\induced{\pmdp}{\sched}}[\val] = 0 \} \]
  containing the states that reach the target with probability zero in
  instantiation $\val$, and the \emph{exist-one states} is the set 
  \[ \poneexstates[\target]{\val} \coloneqq \{ s \in S \mid \exists \sched \in \DMSched{\pmdp}
  \text{ s.t. }\prsol[s]{\target}{\induced{\pmdp}{\sched}}[\val] = 1 \}.
  \]
  The sets $\pzeroallstates[\target]{\val}$ and $\poneallstates[\target]{\val}$ are defined analogously.	
\end{definition}
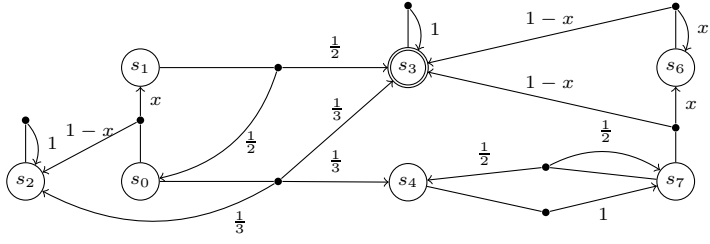
\begin{figure}
\centering
\input{figures/encoding_pmdpexample}
\caption{Example pMDP}	
\label{fig:encoding:pmdpexample}
\end{figure}
\begin{example}
Consider the pMDP in Figure~\ref{fig:encoding:pmdpexample}, with $T = \{ s_3 \}$.
First, assume $\val \coloneqq \{ x \mapsto \nicefrac{1}{2} \}$.
We have \[ \pzeroexstates[\target]{\val} = \{ s_2, s_4, s_7 \}, \text{ and }\pzeroallstates[\target]{\val} = \{ s_2 \}.\]
For $\val \coloneqq \{ x \mapsto 0 \}$, we have \[ \pzeroexstates[\target]{\val} = \{ s_0, s_2, s_4, s_7 \}\text{, and }\pzeroallstates[\target]{\val} = \{ s_2 \},\]
and for $\val \coloneqq \{ x \mapsto 1 \}$, we have \[\pzeroexstates[\target]{\val} = \{ s_2, s_4, s_6, s_7 \}= \pzeroallstates[\target]{\val}.\]
\end{example}
%%
%%
%\paragraph{Nondeterminism}
%For existential strategies, a state has a positive probability to reach the target
%iff there is some action that leads to at least one successor which has a
%positive probability to reach the target and is closer to the target. \todo{T: and what about universal non-determinism? Also this paragraph should maybe go before the example?}
%
\subsubsection{Quantitative analysis}
For pMDPs, we omit the special case of graph-preservation. 
Instead, we consider existential and universal nondeterminism separately.
Contrary to pMCs, but in line with parameter-free MDPs, we now also have to
distinguish properties with lower bounds and properties with upper bounds.

\paragraph{Existential nondeterminism}
Existential nondeterminism is conceptually simple, as we existentially quantify
over both parameter values and strategies.  In a game-theoretic sense, one
player chooses both parameter values and strategies, and we may just
generalise the pMC encoding and use the ETR (where the player selects the
values for all variables).  We may, however, avoid variables for the
strategies by observing that the quantification over strategies is over a
finite set, and that this choice may be represented by a (finite) disjunction.
This disjunction ranges over exponentially many strategies.  We avoid this
explicit blowup by recalling that the nondeterminism is resolved locally.
Instead of a disjunction over all strategies, we make disjunctions over the
local action choices, similar to the encoding of the qualitative case.  These
insights yield a compact encoding, detailed below.
\begin{constraints}[Upper-bounded reachability, existential nondeterminism] 
\begin{align*}
	& p_s \ceq 1 & & \text{for all } s \in \target, \\
	& q_s \text{ is true} & &\text{for all } s \in \target, \\
	& q_s \leftrightarrow \bigwedge_{\act \in \EnAct(s)} \bigvee_{s'\in S} \left( \pprob(s,\act, s') \cgt 0 \rightarrow \left( q_{s'} \land r_{s} \clt r_{s'} \right) \right)	 \quad & &\text{for all } s \in S \setminus \target, \\
	 	& \neg q_s \rightarrow p_s \ceq 0 & & \text{for all } s \in S \setminus \target, \\
	& q_s \rightarrow \bigvee_{\act \in \EnAct(s)} \left( p_s \ceq \sum_{s' \in S} \pprob(s,\act,s') \cdot p_{s'} \right) & &\text{for all } s \in S \setminus \target. 
\end{align*}%	
We refer to the corresponding formula as $\etrau{\pmdp}$.
\end{constraints}
Under existential nondeterminism, we can freely choose the action at every
state: the probability of reaching the target is the sum over the
probabilities of reaching the target from the successors after taking this
action.  As we can choose the action, we thus have a disjunction over
equalities for every state and these disjunctions are guarded by the flag that
the probability is positive from this state, as for pMCs. 

For upper bounds on the reachability probability, under existential
nondeterminism, the strategy tries to minimise the probability.  In
particular, the strategies sets states to probability zero if there is any
strategy to do so.  As the interpretation of the $q_s$ variables is positive
reachability, i.e., $q_s$ is $\true$ iff it is not a zero-state, we obtain a
conjunction over all actions in the encoding.

Below, we give the encoding for lower bounds on the probability.
\begin{constraints}[Lower-bounded reachability, existential nondeterminism] 
\begin{align*}
	& p_s \ceq 1 & & \text{for all } s \in \target, \\
	& q_s \text{ is true} & &\text{for all } s \in \target, \\
	& q_s \leftrightarrow \bigvee_{\act \in \EnAct(s)} \bigvee_{s'\in S} \left( \pprob(s,\act, s') \cgt 0 \rightarrow \left( q_{s'} \land r_{s} \clt r_{s'} \right) \right)	 \quad & &\text{for all } s \in S \setminus \target, \\
	 	& \neg q_s \rightarrow p_s \ceq 0 & & \text{for all } s \in S \setminus \target, \\
	& q_s \rightarrow \bigvee_{\act \in \EnAct(s)} \left( p_s \ceq \sum_{s' \in S} \pprob(s,\act,s') \cdot p_{s'} \right) & &\text{for all } s \in S \setminus \target. 
\end{align*}%	
We refer to the corresponding formula as $\etral{\pmdp}$.
\end{constraints}
For lower bounds, (only) the computation of the zero states changes, as we now
try to avoid setting a state to probability zero. Thus, we only set the
reachability probability to zero if all actions lead to zero-states. Again,
as the interpretation of $q_s$ is that $s$ is not a zero-state, we obtain a
disjunction over the actions.
\begin{theorem}
\label{thm:etrangelic}
 Let $\pmdp$ be a pMDP with target states $\target$. Let $\region \subseteq
 \wdval{\pmdp}$ be a semialgebraic set given by $\psi_R$. For $\bowtie \;\in \{ \leq, <, \geq, > \}$ we define 
     \[ \psi \coloneqq \etra{\pmdp} \land  p_\init \bowtie \nicefrac{1}{2} \land \psi_R. \]
  Then, for any $\val \in \Val$,
	\[
    \val \text{ satisfies } \psi \text{ iff }
    \val \in R \land
    \exists\;\sched \in \DMSched{\pmdp}:\;
    \pr{\instantiated{\pmdp}{\val}}{\sched}{\eventually\target} \bowtie
    \nicefrac{1}{2}.
  \]
\end{theorem}
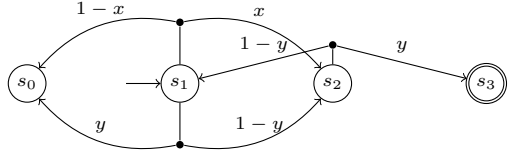
\begin{figure}
\centering
\input{figures/encoding_pmdpexamplesmall}
\caption{Small pMDP to illustrate encodings}	
\label{fig:encoding:pmdpexamplesmall}
\end{figure}
\begin{example}
  Consider the pMDP $\pmdp$ in Figure~\ref{fig:encoding:pmdpexamplesmall} and let $\region \subseteq \wdval{\pmdp}$ be an arbitrary semialgebraic set. We first consider the encoding from Theorem \ref{thm:etrangelic} for an upper bound (i.e., $\bowtie\;=\;\leq)$.
%  for whether $\exists \val \in \region, \exists \sched \colon \pr{\instantiated{\pmdp}{\val}}{\sched}{\eventually\target} \leq  \frac{1}{2}.$
  For conciseness, we simplified several constraints.
	\begin{align*}
		&  p_1 \cleq \nicefrac{1}{2} \land \psi_\region \land p_0 \ceq 0 \land \neg q_0 \land p_3 \ceq 1 \land q_3, \\
		& q_1 \leftrightarrow \Big( \big( x \cgt 0 \land q_2 \land r_1 \clt r_2 \big)  \land \big( 1-y \cgt 0 \land q_2 \land r_1 \clt r_2  \big) \Big), \\
		& q_2 \leftrightarrow \Big( \big( y \cgt 0 \big) \lor \big( 1-y \cgt 0 \land q_1 \land r_2 \clt r_1 \big) \Big), \\
		& q_1 \rightarrow \Big( \big( p_1 \ceq x \cdot p_2 \big) \lor \big( p_1 \ceq (1-y) \cdot p_2 \big)  \Big) \land \neg q_1 \rightarrow p_1 \ceq 0, \\
		& q_2 \rightarrow \big( p_2 \ceq (1-y) \cdot p_1 + y \big) \land \neg q_1 \rightarrow p_2 \ceq 0.
	\end{align*}
	Below, we give the encoding for $\bowtie\;=\;\geq$.
%	whether
%  $\exists \val \in \region, \exists \sched \colon 
%  \pr{\instantiated{\pmdp}{\val}}{\sched}{\eventually\target} \geq
%  \frac{1}{2}.$
  Compared to the encoding above, only the first constraint and one further connective in the second line changed:
  \textcolor{lightgray}{\begin{align*}
		& p_1 \textcolor{black}{\ \cgeq\ } \nicefrac{1}{2} \land \psi_\region \land p_0 \ceq 0 \land \neg q_0 \land p_3 \ceq 1 \land q_3, \\
		& q_1 \leftrightarrow \Big( \big( x \cgt 0 \land q_2 \land r_1 \clt r_2 \big) \textcolor{black}{\ \lor\ } \big( 1-y \cgt 0 \land q_2 \land r_1 \clt r_2  \big) \Big), \\
		&\textcolor{black}{\ldots}
	\end{align*}}
\end{example}

\paragraph{Universal nondeterminism}
For the universal case, we existentially quantify over parameter values and
universally over strategies.  The insight is that the universal quantification
is over a finite domain and may therefore be turned in a conjunction,
analogously to the existential case above.  However, when applying the
conjunction locally at the states, we have to ensure that we do not expect all
equalities to hold simultaneously.  Instead, we adapt the encoding of the
\emph{Bellman inequations} from parameter-free MDPs.  All further ideas are then
straightforward analogues.  Naturally we have to change the zero-states to the
universal case.
\begin{constraints}[Upper-bounded reachability, universal nondeterminism] 
\label{constr:ubreachdnondet}
\begin{align*}
	& p_s \ceq 1 & & \text{for all } s \in \target, \\
	& q_s \text{ is true} & &\text{for all } s \in \target, \\
	& q_s \leftrightarrow \bigvee_{\act \in \EnAct(s)} \bigvee_{s'\in S} \left( \pprob(s,\act,s') \cgt 0 \rightarrow \left( q_{s'} \land r_{s} \clt r_{s'} \right) \right)	 \quad & &\text{for all } s \in S \setminus \target, \\
	 	& \neg q_s \rightarrow p_s \ceq 0 & & \text{for all } s \in S \setminus \target, \\
	& q_s \rightarrow \bigwedge_{\act \in \EnAct(s)} \left( p_s \cgeq \sum_{s' \in S} \pprob(s,\act, s') \cdot p_{s'} \right) & &\text{for all } s \in S \setminus \target.
\end{align*}%
We refer to the corresponding formula as $\etrdu{\pmdp}$.	
\end{constraints}
For lower-bounded reachability, it suffices to
change the zero-state computation and the inequalities on the probabilities.
We refer to the corresponding formula as $\etrdl{\pmdp}$.
The accompanying encoding is then similar to the existential case.
\begin{theorem}
 Let $\pmdp$ be a pMDP with target states $\target$. Let $\region \subseteq
 \wdval{\pmdp}$ be a semialgebraic set given by $\psi_R$.
 For $\bowtie \;\in \{ \leq, <, \geq, > \}$ we define  
 \[ \psi \coloneqq \etrd{\pmdp} \land  p_\init \bowtie \nicefrac{1}{2} \land \psi_R \]
% \begin{itemize}	
%   \item Either $\bowtie \;\in \{ \leq, < \}$, and we define 
%     \[ \psi \coloneqq \etrdu{\pmdp} \land  p_\init \bowtie \lambda \land \psi_R, \]
%   \item or $\bowtie \;\in \{ \geq, > \}$, and we define
%     \[ \psi \coloneqq \etrdl{\pmdp} \land  p_\init \bowtie \lambda \land \psi_R. \]
%  \end{itemize}
  Then, for any $\val \in \Val$,
	\[
    \val \text{ satisfies } \psi \text{ iff }
    \val \in R \land
    \forall\;\sched \in \DMSched{\pmdp}:\;
    \pr{\instantiated{\pmdp}{\val}}{\sched}{\eventually\target} \bowtie
    \nicefrac{1}{2}.
  \]
\end{theorem}

\subsection{Lower bounds}
Following the results from the previous sections, we have an \ETR{} upper bound.
For Boolean valuations, we even have an \NP{} upper bound by guessing which parameters are assigned to one. 
Therefore, the lower bound for the qualitative and the upper bound for the quantitative case coincide.

We first reduce some entries in the table to each other. 
The remainder of this section then first considers hardness results for the case of arbitrary parameters, and then shows better upper bounds for the case where the number of parameters is fixed.

\paragraph{Considerations for the comparison relations}
The number of combinations that we need to consider is significantly reduced
by a couple of reductions that follow from structural properties of pMCs.
Intuitively, the first property is based on the duality of target and bad
states:
\begin{proposition} \label{pro:complexity:restrictrelations}
  For every $\quant \in \{\exists,\forall\}$, 
  there are polynomial-time many-one reductions
  \begin{itemize}
    \item among the problems $\quant \reachdp{>}_{\mathrm{gp}}$ and
      $\quant \reachdp{<}_{\mathrm{gp}}$ and
    \item among the problems $\quant \reachdp{\geq}_{\mathrm{gp}}$ and
      $\quant \reachdp{\leq}_{\mathrm{gp}}$.
  \end{itemize}	
\end{proposition}
\begin{proof}
  We prove only the first item for $\quant = \exists$. All other cases may be
  proven analogously.  First, we deduce from \cite[Thm.~10.122 and
  Thm.~10.127]{BK08} that, in polynomial time, and without regarding the actual
  transition probabilities, we can compute from $\pmdp$ and a target set $T$,
  a target set\footnote{Which is some adequate union of particular maximal end
  components in $\pmdp$.}  $T'$ such that for each $\mdp \in
  \generator[\mathrm{gp}]{\pmdp}$: 
  \begin{align*}
    \max_{\sched \in \DMSched{\pmdp}} \pr{\mdp}{\sched}{\eventually\target} = 1 - \min_{\sched
    \in \DMSched{\pmdp}} \pr{\mdp}{\sched}{\eventually \target'}.
  \end{align*}
  Please observe that the step above in general does not work without the
  restriction to graph-preserving instantiations. We combine this to obtain:
  \begin{align*}
    & \exists\;\val \in \gpval{\pmdp},\;
    \exists\;\sched \in \DMSched{\pmdp}:\;
    \pr{\mdp[\val]}{\sched}{\eventually\target} > \frac{1}{2}\\
    \iff & 
    \exists\; \val \in \gpval{\pmdp}:\;
    \max_{\sched \in \DMSched{\pmdp}}
    \pr{\mdp[\val]}{\sched}{\eventually\target} > \frac{1}{2}\\
    \iff &
    \exists\; \val \in \gpval{\pmdp}:\;
    \left(1 - \min_{\sched \in \DMSched{\pmdp}}
           \pr{\mdp[\val]}{\sched}{\eventually\target'}\right) > \frac{1}{2}\\
    \iff &
    \exists\; \val \in \gpval{\pmdp}:\;
    \min_{\sched \in \DMSched{\pmdp}} 
     \pr{\mdp[\val]}{\sched}{\eventually\target'} < \frac{1}{2}\\
    \iff &
    \exists\; \val \in \gpval{\pmdp},
    \exists\; \sched \in \DMSched{\pmdp}:\;
     \pr{\mdp[\val]}{\sched}{\eventually\target'} < \frac{1}{2}.
  \end{align*}
\end{proof}
For strict lower-bounded reachability, we can restrict our attention to
graph-preserving parameter instantiations.
\begin{proposition}
\label{pro:e_reach_less_red_e_reach_less_gp}
$\reachdp{>}_\mathrm{wd} $ is polynomially reducible to $ \reachdp{>}_{\mathrm{gp}}$.
\end{proposition}
This proposition is an immediate consequence of the semi-continuity of the solution
function for simple pMCs~\cite[Thm.~5]{DBLP:conf/uai/Junges0WQWK018}.
Conversely, we can also construct gadgets that avoid valuations which are not
graph-preserving.  Using the gadget in Figure~\ref{fig:e_reach_gp_red_e_reach}
we can ensure that for any non graph-preserving instantiation, the probability to
reach the target is $0$, while the reachability probabilities
for graph-preserving instantiations are not affected. Together with semi-continuity of the
solution function, we deduce:
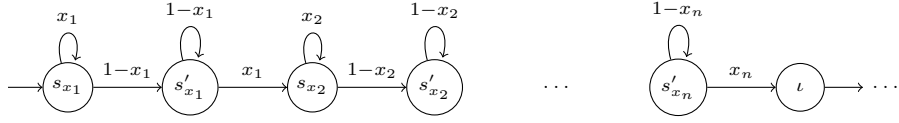
\begin{figure}
\centering
\input{figures/complexity_gpvswdequiv}
\caption{Gadget for the reduction from Prop.~\ref{pro:e_reach_gp_red_e_reach}. $\init$ is the initial state of the given pMC.}
\label{fig:e_reach_gp_red_e_reach}
\end{figure}
\begin{proposition} \label{pro:e_reach_gp_red_e_reach}
  $\reachdp{>}_{\mathrm{gp}} $ is polynomially reducible to
  $\reachdp{>}_\mathrm{wd}$, and similarly $\reachdp{\geq}_{\mathrm{gp}} $ is
  polynomially reducible to $\reachdp{\geq}_\mathrm{wd}$.
\end{proposition}
\begin{proof}
Let $\pmc$ be a simple pMC. We extend $\pmc$ with the gadget outlined in Figure~\ref{fig:e_reach_gp_red_e_reach}. Formally, we construct a pMC $\pmc'$ with states $S'\coloneqq S \cup \{ s_x, s'_x \mid x \in \pars \}$, initial state $s_{x_1}$ and
 \begin{align*} \pprob'(s,s') \coloneqq \begin{cases}
 \pprob(s,s) & \text{if }s,s' \in S, \\
 x    & \text{if }s = s' = s_x, \\
 1{-}x & \text{if }s = s' = s'_x, \\
 1{-}x & \text{if }s = s_x \text{ and } s' = s'_x, \\
 x & \text{if }s = s'_x \text{ and } s' = \text{next}(s'_{x}), \\
 0 & \text{otherwise.}	
 \end{cases}
 \end{align*}
where $\text{next}(s'_x)$ is $s_{x+1}$ if $x = x_i$ for some $i < |\pars|$, and $\init$ if $i = |\pars|$, where $\init$ is the initial state of $\pmc$.
The pMC $\pmc'$ is only linearly larger than $\pmc$. 
Observe that the construction of the gadget may be adapted for non-simple pMCs
(with different well-defined parameter valuations).
By construction, we have for every $\val \in \wdval{\pmc}$ and $T \subseteq S$:
 \begin{align*} \pr{\instantiated{\pmc'}{\val}}{}{\eventually\target} =  \pr{\instantiated{\pmc'}{\val}}{}{\eventually\{ \init \}}  \cdot \pr{\instantiated{\pmc}{\val}}{}{\eventually\target}.  \end{align*}
We observe the following: 
\begin{align*}& \forall\;\val \in \gpval{\pmc}:
  \pr{\pmc'[\val]}{}{\eventually\{ \init \}} = 1 \text{~~and thus~~} \pr{\instantiated{\pmc'}{\val}}{}{\eventually\target} =
  \pr{\instantiated{\pmc}{\val}}{}{\eventually\target}
%  \implies & \forall\; \val \in \gpval{\pmdp} \setminus \wdval{\pmdp}:
%  \pr{\instantiated{\pmc}{\val}}{}{\eventually\{ \init \}} = 0.
\end{align*}
and
\begin{align*}
	\forall\; \val \in \wdval{\pmdp} \setminus \gpval{\pmdp}:
	\pr{\instantiated{\pmc'}{\val}}{}{\eventually\{ \init \}} = 0 \text{~~and thus~~} \pr{\instantiated{\pmc'}{\val}}{}{\eventually\target} =
	0.
\end{align*}
Together, we deduce:
\[
  \exists\;\val \in \gpval{\pmc}: \pr{\instantiated{\pmc}{\val}}{}{\eventually\target} \unrhd \nicefrac{1}{2}
	\quad\iff\quad
	\exists\;\val \in \wdval{\pmc'}: \pr{\instantiated{\pmc'}{\val}}{}{\eventually\target} \unrhd \nicefrac{1}{2}.
\]
\end{proof}
\noindent We are not aware of any such reductions for upper bounds.

\subsubsection{pMCs with arbitrarily many parameters}
We first consider the upper right part of Table~\ref{tab:complexity}:
reachability in pMCs with an unbounded number of parameters.

\paragraph{Non-strict inequalities}
First, we establish the following theorem.
\begin{theorem} \label{thm:etr:pmcsarehard}
  The $\reachdp{\leq}_{*}$ and $\reachdp{\geq}_{*}$ problems are all
  \ETR-complete even for acyclic pMCs.
\end{theorem}
\noindent For this result, we reduce from the following \ETR-hard problem.
\begin{definition}
\label{def:etr:mb4feasc}
  The \emph{modified-closed-bounded-4-feasibility} (mb4FEAS-c) problem asks: 
  Given a (non-negative) polynomial $f$ of degree 4, does there exist some
  \( \val\colon \pars \rightarrow [0,1] \) such that \(f[\val] \leq 0\)?
  The \emph{modified-open-bounded-4-feasibility} (mb4FEAS-o) problem is analogously
  defined with $\val$ ranging over $(0,1)$. 
\end{definition}
\noindent This problem easily reduces to its $\geq$-variant by multiplying $f$
with $-1$.
\begin{lemma}
\label{lem:etr:mb4feas}
  The mb4FEAS-c and mb4FEAS-o problems are \ETR-hard.
\end{lemma}
\begin{proof}[Proof sketch]
  Essentially, one reduces from the existence of common roots of quadratic
  polynomials lying in a unit ball, which is 
  \ETR-complete~\cite[Lemma 3.9]{Schaefer2013}. The reduction to mb4FEAS
  follows the reduction\footnote{Essentially the polynomial $f$ in mb4FEAS is
  constructed by taking the sum-of-squares of the quadratic polynomials, and
  further operations are adequately shifting the polynomial.} between
  unconstrained variants (i.e., variants in which the position of the root is
  not constrained) of the same decision problem~\cite[Lemma
  3.2]{DBLP:journals/mst/SchaeferS17}.
\end{proof}

Before presenting a proof of our \ETR-hardness claim we recall the following
result hinted at by Chonev~\cite{DBLP:conf/rp/Chonev19}.
More precisely, we consider the question: Given a polynomial $f$, does there exist a (simple, acyclic) pMC such that $\prsol{\target}{\pmc} = f$? 
We start with a positive example.
\begin{figure}
\centering
	\input{figures/aspects_frompoltopmc}
	\caption{pMC for the polynomial $\frac{1}{2}x^2 + \frac{1}{3}y$}
	\label{fig:aspects:frompoltopmc}
\end{figure}
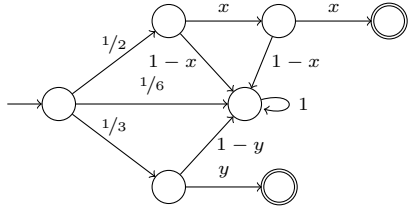
\begin{example}
	The polynomial $f = \frac{1}{2}x^2 + \frac{1}{3}y$ corresponds to the solution function of the pMC in Figure~\ref{fig:aspects:frompoltopmc}.
\end{example}
The polynomial in the example is easy to translate. In particular, all coefficients are positive and they sum up to a value less than one.
In the pMC, all transitions of the form $1-x$ (for any parameter $x$) go to the sink state immediately. 
To handle negative coefficients, we are going to make a more flexible use of the $1-x$ transitions. 
We first reformulate the polynomials.
\begin{lemma}(\cite[Remark~1]{DBLP:conf/rp/Chonev19})
\label{lem:aspects:polyreformulation}
Let $f \in \QQ[\vars]$ be a polynomial. We can rewrite $f$ as:
	\begin{align}
		\label{eq:chonevtrickintermediate}
		f = \sum_{i=1}^m a_i \cdot h_i + b\qquad\text{ with }h_i =
 \prod_{1 \leq j \leq \nrvars} x_j^{e_{i,j}}\cdot(1-x_j)^{e'_{i,j}}
	\end{align}
	 with  $a_i \in \QQpos$, $e_{i,j},e'_{i,j} \in \NN$, and $b \in \mathbb{Q}$.
\end{lemma}
\begin{proof}
	Observe that a monomial $-x_1 \cdot \dots \cdot x_d$ of degree $d \geq 0$ may be written as
	\begin{align}
		\label{eq:chonevmonomialtrick}
		-x_1 \cdot \dots \cdot x_d = {-}1 + \sum_{i=1}^d(1-x_i) \cdot x_{i+1} \cdot \dots \cdot x_d ,
	\end{align}
	which is proved by induction on $d$: For $d = 0$, both sides are ${-}1$ (an empty product equals $1$). For $d \geq 0$, we multiply both sides of \eqref{eq:chonevmonomialtrick} by $x_{d+1}$ to obtain
	\begin{align*}
		 -x_1 \cdot \dots \cdot x_d \cdot x_{d+1} 
		&= {-} x_{d+1} + \sum_{i=1}^{d}(1-x_i) \cdot x_{i+1} \cdot \dots \cdot x_d \cdot x_{d+1}  \\
		&= (1 - x_{d+1}) - 1  + \sum_{i=1}^{d}(1-x_i) \cdot x_{i+1} \cdot \dots \cdot x_d \cdot x_{d+1} \\
		&= {-1} + \sum_{i=1}^{d+1}(1-x_i) \cdot x_{i+1} \cdot \dots \cdot x_d \cdot x_{d+1}.
	\end{align*}	Hence applying \eqref{eq:chonevmonomialtrick} to every term of
  $f$ we obtain Equation~\eqref{eq:chonevtrickintermediate} 
%	\begin{align*}
%		f = \sum_{i=1}^m a_i \cdot h_i + b
%	\end{align*}
	 where the $a_i$ are \emph{positive} rational coefficients, the $h_i$  are nonempty products of terms from $\{x, (1-x) \mid x \in \pars \}$ and $b \in \mathbb{Q}$ is a constant term.
	  We may assume that $b \leq 0$, otherwise $b = b \cdot x + b \cdot (1-x)$
    for any $x \in \pars$ and we may ``pull'' $b$ inside the sum.	
\end{proof}
We will show that this reformulation
allows to translate and scale a polynomial $f$ such that there exists a pMC $\pmc$ with targets $\target$ and	\[  \frac{f + A}{B} = \prsol{\target}{\pmc} \quad\text{ for some }A \in \QQnn, B \in \QQpos \]

  \begin{figure}
  \centering
\input{figures/aspects_chonev_trick}
  \caption{pMC with $\prsolgp{\target}{\pmc} = \nicefrac{{-}2x^2y + y + 2}{8}$}
  \label{fig:chonevtrick}
  \end{figure}
\begin{example}
	Consider the polynomial ${-}2x^2y + y$. 
	We reformulate this to:
	\[ 2\cdot\left( (1-x)xy + (1-x)y + (1-y) - 1 \right) + y \] and then to \[ 2
	\cdot (1-x)xy + 2\cdot (1-x)y + 2\cdot(1-y) + y - 2.\]  
	 After shifting upwards (with $+2$) and rescaling (with $\frac{1}{8}$), we can construct the pMC $\pmc$ depicted in Figure~\ref{fig:chonevtrick}.
	\label{ex:chonevtrickpols}
\end{example}

\noindent
Formally, we show the following slightly more general proposition. 
\begin{proposition}(\cite{DBLP:conf/rp/Chonev19})
\label{pro:aspects:chonevrescaling}
	Let $f$ be a polynomial. For any $A$ and $B$ sufficiently large, there exists a pMC $\pmc$ with targets 
	$\target$ such that 
	\[  \frac{f + A}{B} = \prsol{\target}{\pmc}. \]
	
  Moreover, if $d$ is the total degree of $f$, $t$ the number of terms in $f$
  and $\kappa$ a bound on the (bit-)size of the coefficients and the thresholds $\mu$,
  $\lambda$, then $\pmc$ may be
  constructed in time $\mathcal{O}(poly(d,t,\kappa))$.	
\end{proposition}

\begin{proof}
Recall that $f$ may be written as  
\[ 
f = \sum_{i=1}^m a_i \cdot h_i + b \quad\text{ with $a_i \in \QQnn$ and $b \in \QQneg$. }
 \]
 Let $A > b$. We reformulate: 
 \[ 
 f + A = \sum_{i=1}^m a_i \cdot h_i + b' \quad\text{ with $a_i \in \QQnn$ and $b' = (b+A) \in \QQnn$.}
 \]
 Let $B > \sum_{i=1}^m a_i + b'$.
We can write 
\[ 
\tilde{f} \coloneqq  \frac{f + A}{B} = \sum_{i=1}^m \tilde{a_i} \cdot h_i + \tilde{b}
\]  with $\tilde{a}_i, \tilde{b} \in \QQnn$ and  $\sum_{i=1}^m \tilde{a}_i + \tilde{b} < 1$.
	The modified polynomial $\tilde{f}$ naturally corresponds to a simple acyclic pMC $\tilde{\pmc}$ with $\prsol{\target}{\pmc} = \tilde{f}$ as shown in Figure~\ref{fig:chonevtrickidea}.
	\begin{figure}[tb]
	\centering
	\input{figures/aspects_chonev_trick_proof_idea}
	\caption{The essential construction of the pMC in Proposition~\ref{pro:aspects:chonevrescaling}: Any probability mass not drawn goes to a sink. }
	\label{fig:chonevtrickidea}
	\end{figure}
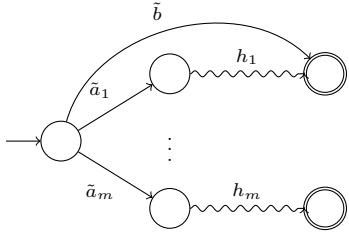

	For the complexity of the construction, notice that $m$ in the sum \eqref{eq:chonevtrickintermediate} is in $\mathcal{O}(td)$ where $d$ and $t$ are bounds on the total degree and the number of terms of $f$, respectively. 
	The $h_i$ are products of at most $d$ terms. 
	The $a_i$ are the absolute values of the original coefficients of $f$ and $b$ is the sum of at most $t$ of those coefficients. 
	Hence $a_i$, $b$, $A$, $B$ and the polynomial $\tilde{f}$ may be computed in time $\mathcal{O}(poly(t,d,\kappa))$. The same then also holds for the pMC $\tilde{\pmc}$.
\end{proof}	
%
%\begin{lemma}[Chonev's trick~{\cite[Remark~1]{DBLP:conf/rp/Chonev19}}]
%  \label{lem:chonevtrick}
%  Let $f \in \mathbb{Q}[X]$ be a polynomial, $\mu \in \QQ$ and any $0 < \lambda <
%  1$. There exists a simple acyclic pMC $\mdp$ with a target state $T$ such that
%  for all $\val \colon X \rightarrow [0,1]$ and all comparison relations
%  $\bowtie~\in \{<,\leq,\geq,>,=\}$ it holds that
%	\[ 
%    f[\val] \bowtie \mu \Longleftrightarrow \Pr_{\mdp[\val]}(\lozenge T) \bowtie \lambda.
%	\]
%  Moreover, if $d$ is the total degree of $f$, $t$ the number of terms in $f$
%  and $\kappa$ a bound on the (bit-)size of the coefficients and the thresholds $\mu$,
%  $\lambda$, then $\mdp$ can be
%  constructed in time $\mathcal{O}(poly(d,t,\kappa))$.
%\end{lemma}

\begin{proof}[Proof of Theorem~\ref{thm:etr:pmcsarehard}]
  The reduction from mb4FEAS-c to $\reachdp{\leq}_\mathrm{wd}$ consists in
  constructing for a given polynomial a pMC using Propostion~\ref{pro:aspects:chonevrescaling}
  with $\mu = 0$ and $\lambda = \frac{1}{2}$.  For
  $\reachdp{\leq}_\mathrm{gp}$, we reduce from the open variant
  and notice that as the construction in Propostion~\ref{pro:aspects:chonevrescaling} preserves
  all satisfying instantiations $\val\colon X \rightarrow [0,1]$ it, in
  particular, also preserves them on the graph-preserving parameter
  valuations.  For $\geq$, we apply Proposition~\ref{pro:aspects:chonevrescaling} on ${-}f$. 
\end{proof}
Observe that there are instances of the mb4FEAS problems which admit a unique
solution, and this solution may be irrational. In contrast, if there exists a
solution for a constraint $f \cgt 0$, then there exist infinitely many
(rational) solutions. To the best of our knowledge, the complexity of a
variant of these problems with strict bounds is open. Therefore, we have no
\ETR-hardness proof for $\reachdp{}$ with strict bounds.  In general,
\emph{conjunctions} of strict inequalities are also \ETR-complete
\cite{DBLP:journals/mst/SchaeferS17}. We exploit this in the proof of
Theorem~\ref{thm:etr:mdpsarehard}.

\paragraph{Strict inequalities}
We now move to the problems with strict inequalities.
%%%
\begin{theorem} \label{thm:e_reach_gr_gp_np_hard}
	$\reachdp{>}_{\text{*}}$ and $\reachdp{<}_{\text{*}}$ are \NP-hard.
\end{theorem}
Using Propositions~\ref{pro:e_reach_gp_red_e_reach}
and~\ref{pro:e_reach_less_red_e_reach_less_gp}, we may  restrict our attention
to well-defined parameter valuation sets.  Recall
Proposition~\ref{pro:complexity:qualitativereach}: Almost-sure reachability is
\NP-hard.  A more refined analysis of the 3SAT-reduction yields:
\begin{proposition} \label{pro:e_reach_gr_np_hard}
	$\reachdp{>}_\mathrm{wd}$ and $\reachdp{<}_\mathrm{wd}$ are \NP-hard.
\end{proposition}
\begin{proof}[Proof sketch]
	Reconsider the construction in Figure~\ref{fig:aspects:chonevconstruction}. 
  We first show the following claim to simplify our proof afterwards:

  \item
  \paragraph{Auxiliary claim} If $\psi$ is unsatisfiable, then for all $\val \in
  \wdval{\pmc}$ there exists some clause $c_{i^*}$, such that
  $\pprob(l_{i^*,j}, \bot)[\val] \geq \frac{1}{2}$ for all $j \in\{1,2,3\}$,
  or more formally
  \begin{equation} \label{eq:auxclaim}
  \begin{aligned}
    & \psi \text{ is unsatisfiable}\\
    \implies &\forall\;\val \in \wdval{\pmc},\;\exists i^* \in \{1,\hdots,m\}, \forall j \in\{1,2,3\}:\; \pprob(l_{i^*,j}, \bot)[\val] \geq \frac{1}{2}.
  \end{aligned}
  \end{equation}

  \item
  \paragraph{Proof of the auxiliary claim} 
  Let $\psi$ be satisfiable and assume towards contradiction that for some $\val$ and for every clause $c_i$ there is a `witness' literal $l_{i,j}$ with $\pprob(l_{i,j}, \bot)[\val] < \frac{1}{2}$.
  Together with the definition of $\pprob$, we conclude either
  \begin{enumerate}
    \item $l_{i,j}$ is a variable $x$, and $\val(\tilde{x}) > \frac{1}{2}$ or
    \item $l_{i,j}$ is a negated variable $\overline{x}$, and $\val(\tilde{x}) < \frac{1}{2}$.
  \end{enumerate}

  We now construct a satisfying assignment for $\psi$:
  Consider an assignment $\val_\psi$ for $\psi$, with
  \[ \val_\psi(x) \coloneqq
  \begin{cases}
    \true, & \text{if } \val(\tilde{x}) > \frac{1}{2}, \\
    \false, & \text{if } \val(\tilde{x}) < \frac{1}{2}, \\
    \text{arbitrary}, & \text{if } \val(\tilde{x}) = \frac{1}{2}.
  \end{cases}\]
  In both case 1 and 2 above, $\val_\psi$ satisfies clause $c_i$. 
  Thus $\psi$ is satisfiable, contradiction.
  \medskip

  \newcommand{\prfromto}[2]{\pr[{#1}]{\dtmc}{}{\lozenge^{>0}\; #2}}
  \item
  \paragraph{Proof for correctness of reduction} 
  We only show:
  \begin{equation}
    \label{eq:chonevpart2}
    \psi \text{ is unsatisfiable} \iff \forall\;\val \in \wdval{\pmc}: \pr{\instantiated{\pmc}{\val}}{}{\eventually\target} \leq \frac{2}{3},
  \end{equation}
  which is equivalent to:
  \[
    \psi\text{ is satisfiable} \iff \exists\;\val \in \wdval{\pmc}: \pr{\instantiated{\pmc}{\val}}{}{\eventually\target} > \frac{2}{3}.
  \]
  Again let $\psi$ be unsatisfiable and fix a parameter valuation $\val$ and set $\dtmc \coloneqq \instantiated{\pmc}{\val}$. We show  $\pr{\dtmc}{}{\eventually\target} \leq \frac{2}{3}$.
  Let $i^*$ be like in the auxiliary claim \eqref{eq:auxclaim}. The idea here is that $c_{i^*}$ is the (potentially only) unsatisfied clause. 
  By construction of $\pmc$ and the auxiliary claim,
  \[
  \prfromto{l_{i^*,j}}{v_k} \leq 1 - \pprob(l_{i*,j}, \bot)[\val] \leq \frac{1}{2}
  \]
  for all $j \in \{1,2,3\}$. Hence
  \begin{align*}
    \prfromto{c_{i^{*}}}{v_k} = \sum_{j=1}^3 \val(y_{i^*, j})\cdot \prfromto{l_{i^*,j}}{v_k} \leq \frac{1}{2}\sum_{j=1}^3 \val(y_{i^*, j}) = \frac{1}{2}.
  \end{align*}
  Consequently, for $\prfromto{v_k}{v_k}$ it holds that
  \begin{align*}
    \prfromto{v_k}{v_k} &= \pprob(v_k, c_{i^{*}}) \cdot \prfromto{c_{i^{*}}}{v_k}\\
    & \quad \quad \quad \quad \quad \quad \quad \quad + \sum_{i \neq i^*}\pprob(v_k, c_i)\cdot \prfromto{c_i}{v_k} \\
    & \leq \frac{1}{m+1}\cdot\frac{1}{2} + \frac{m-1}{m+1} = \frac{2m-1}{2(m+1)}.
  \end{align*}
  Plugging this into the equation
  \begin{align*}
    \prfromto{v_k}{T} = \frac{1}{m+1} + \prfromto{v_k}{v_k}\cdot \prfromto{v_k}{\target}
  \end{align*}
  yields 
  $\prfromto{v_k}{T} \leq \frac{2}{3}$.
  All paths from $v_0$ to $T$ go through $v_k$, thus: 
  \begin{align*}
  \pr{\dtmc}{}{\eventually \target} = \prfromto{v_0}{v_k} \cdot \prfromto{v_k}{\target} \leq \prfromto{v_k}{\target} \leq \frac{2}{3}.
  \end{align*}

  The remainder of the proof is analogous to the proof of Proposition~\ref{pro:complexity:qualitativereach}.
  The proof for threshold $\nicefrac{1}{2}$ follows by applying the argument
  sketched on page \pageref{page:fixedthreshold} 	
\end{proof}

\subsubsection{pMDPs with arbitrarily many parameters}
We now move to the lower-right corner of Table~\ref{tab:complexity}, and
consider pMDPs without bound on the number of parameters.  For the results, we
distinguish whether the quantifier over the strategies is existential
or universal.

\paragraph{Existential nondeterminism} We remove the nondeterminism by
reducing to pMCs with additional variables for the nondeterminism. (Recall
such local resolution of the nondeterminism is valid because of
Proposition~\ref{pro:detmemless-suff}.)  This
reduction however requires an arbitrary range for the parameters.  More
formally, we obtain:
\begin{proposition} \label{pro:existwdgeeqexistwdge}
  There are polynomial-time many-one reductions among 
  $\reachdp{\bowtie}_\mathrm{wd}$ and $\exists\reachdp{\bowtie}_\mathrm{wd}$.
\end{proposition}
\noindent Minor adaptions to Proposition~\ref{pro:detmemless-suff} and Proposition~\ref{pro:e_reach_gp_red_e_reach}
yield:
\begin{proposition}
	\label{pro:eeggpeqerggwd}
There are polynomial-time many-one reductions among the problems
  $\exists\reachdp{>}_\mathrm{gp}$ and $\reachdp{>}_\mathrm{wd}$.
\end{proposition}

\paragraph{Universal nondeterminism}
We now consider universal nondeterminism. Contrary to pMCs, we obtain \ETR-completeness for pMDPs
and any comparison relation:
\begin{theorem} \label{thm:etr:mdpsarehard}
	$\forall\reachdp{\bowtie}_{*}$ are all \ETR-complete 
  even for acyclic pMDPs.
\end{theorem}
Non-strict relations are already trivially \ETR-hard via Theorem~\ref{thm:etr:pmcsarehard}.
For the strict relations, we reduce from the following problem.
\begin{definition}
\label{def:etr:bcon4ineq}
  The \emph{bounded-conjunction-of-inequalities} (bcon4INEQ-c) problem asks:
  Given a family of polynomials $f_1,\hdots, f_m$ of degree 4, does there
  exist some
  \( \val: \pars \to  [0,1] \text{ such that }\bigwedge_{i=1}^m f_i[\val] < 0
  \)?
  The open variant
  (bcon4INEQ-o) may be defined analogously.
\end{definition}

By a straightforward reduction from mb4FEAS (adapted from~\cite[Thm
4.1]{DBLP:journals/mst/SchaeferS17}) we obtain that:
\begin{lemma} \label{lem:etr:bcon4INEQhard}
	The bcon4INEQ-o/c problems are \ETR-hard.
\end{lemma}

\begin{figure}
\centering
\input{figures/aspects_mdps_etrhard}
\caption{Construction for the proof of Theorem~\ref{thm:etr:mdpsarehard}}
\label{fig:mdpsareetrhard}
\end{figure}
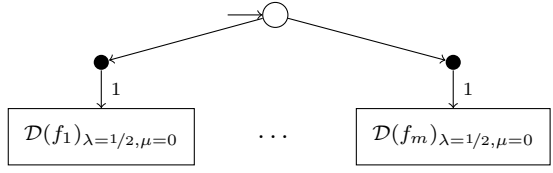

\begin{proof}[Proof of Theorem~\ref{thm:etr:mdpsarehard}]
We show the reduction from the bcon4INEQ problems to $\forall\reachdp{>}_\mathrm{wd}$.
For given $f_1,\hdots,f_m$, we construct pMCs \[
\pmc(f_1)_{\lambda=\nicefrac{1}{2},\mu=0}, \hdots,
\pmc(f_m)_{\lambda=\nicefrac{1}{2},\mu=0} \text{ with target states $T_i$ } \]
by applying Proposition~\ref{pro:aspects:chonevrescaling} to $f_i$ (with $\lambda=\frac{1}{2}$
and $\mu = 0$).
Then, we construct a pMDP as outlined in Figure~\ref{fig:mdpsareetrhard}. We take the disjoint union of the pMCs and adding a fresh initial state, with nondeterministic actions into each pMC.
Formally, let ${\pmc(f_i)_{\lambda=\nicefrac{1}{2},\mu=0} = (S_i,\init^i,\pars, \pprob_i)}$.
We construct a pMDP $\pmdp \coloneqq (S, \init, \Act, \pars, \pprob)$ with
\[ S \coloneqq \bigcup S_i \cup \{ s_0 \}, \init \coloneqq s_0, \Act \coloneqq \{ \act_i \mid 1 \leq i \leq m \}, \] 
and $\pprob$ given by:
\begin{align*}
\pprob(s,\act,s') \coloneqq \begin{cases}
 \pprob_i(s,s') &\text{if } s,s' \in S_i, \act = \act_i\text{ for some }i,\\
 1 & \text{if }s=s_0, s'=\init^i, \act=\act_i\text{ for some }i,\\
 0 & \text{otherwise.} 	
 \end{cases}
\end{align*}
We consider target states $T \coloneqq \bigcup T_i$.
The construction is in polynomial time. 
The pMDP $\pmdp$ has $m$ strategies $\sched_1,\hdots,\sched_m$ with $\sched_i \coloneqq \{ s_0 \mapsto \act_i \}$ (all other states have trivial nondeterminism).

By construction, there exists $\val \in \wdval{\pmdp}$ such that:
 \[
\pr{\mdp[\val]}{\sched_i}{\eventually\target} < \frac{1}{2}
\quad\text{iff}\quad
f_i[\val] < 0.
\]
Then, \[ \exists\;\val \in \wdval{\pmdp}:\; \bigwedge_i \pr{\instantiated{\pmdp}{\val}}{\sched_i}{\eventually\target} < \frac{1}{2}\quad\text{iff}\quad \exists\; \val \in [0,1]^\pars \; \bigwedge_i f_i[\val] < 0,  \]
or equivalently, 
\[ \exists\; \val \in \wdval{\pmdp},\; \forall \sched \in \DMSched{\pmdp}:\; \pr{\instantiated{\pmdp}{\val}}{\sched_i}{\eventually\target}  < \frac{1}{2}\quad\text{iff}\quad \exists\; \val \in [0,1]^\pars \; \bigwedge_i f_i[\val] < 0.  \]
\end{proof}

\subsection{Upper bounds with a fixed number of parameters}
While the \ETR-completeness may be considered bad news, as it renders the problem intractable in general,
there is also good news. In particular, for any fixed number of parameters, the (parametric) complexity is lower.

In our considerations, we focus
on graph-preserving instantiations, as the analysis of pMDP $\pmdp$ and
$\wdval{\pmdp}$ corresponds to analysing constantly many pMDPs on
$\gpval{\pmdp'}$ --- see Lemma~\ref{lem:decomposingwdtogp}.

\begin{theorem}[From~\cite{baiercomplexity,unpublished:infocomp}]
 \label{thm:aspects:ptimepmcs}
 For any fixed number $K$, given a pMC $\pmc$ with at most $K$ parameters,
 determining whether there is a $\dtmc \in \generator{\pmc}$ such that 
 $\pr{\dtmc}{}{\eventually\target} \bowtie \lambda$ is in \Ptime.
\end{theorem}
\noindent That is, in the fixed parameter case, $\reachdp{\bowtie}_{*}$ is in \Ptime.
\paragraph{pMDPs}
With the positive result for pMCs in place, we turn our attention to pMDPs. 
However, we can no longer simply eliminate all state variables in the ETR encoding.
Observe that a reduction from pMDPs with existential nondeterminism to pMCs does not work: it requires the introduction of additional parameters (depending on the number of states).
Indeed, the precise complexity for the problem remains open.
Below, we establish \NP-membership for all variants.

For
pMDPs with existential nondeterminism, \NP-membership is straightforward.
\begin{proposition}
In the fixed parameter case, $\exists\reachdp{\bowtie}_{*}$ is in \NP.
\label{pro:fp_ee_reach_in_np}
\end{proposition}
\begin{proof}
Guess a memoryless strategy. The strategy can be stored using polynomially
many bits\footnote{contrary to guessing parameter values, as they are real
numbers.}.  Construct the induced pMC, and verify it in \Ptime.
\end{proof}

For pMDPs with universal nondeterminism, \NP-membership is more involved. 
\begin{theorem}
In the fixed parameter case, $\forall\reachdp{\bowtie}_{*}$ is in \NP.
\label{thm:fp_ea_reach_in_np}
\end{theorem}
%The following lemma is not strictly necessary, but hints at the solution.
%\begin{lemma}
%	In the fixed parameter case, checking whether a given strategy is somewhere optimal on $\wdval{\pmdp}$ is in \Ptime.
%\end{lemma}
%\begin{proof}
%The ETR formula in Constraints~\ref{const:aspects:somewhereminimalratfunc} only has $|\pars|$ many variables (and is polynomially larger than the pMDP).	
%\end{proof}
The essential trick for \NP-membership  for universal nondeterminism is
guessing an optimal strategy and verifying the induced pMC together with
checking that the strategy is indeed optimal. For the verification step, we
will make use of the following ETR encoding based on the Bellman optimality
equations for minimising strategies in parameter-free MDPs. For conciseness,
we give it here only for graph-preserving valuations.
\begin{constraints} \label{const:aspects:somewhereminimalratfunc}
Let $\pmdp$ be a pMDP and consider a set of valuations $\region \subseteq \gpval{\pmdp}$. 
Let $\sched \in \DMSched{\pmdp}$ and let 
$\nicefrac{h_s}{g_s} \coloneqq \prsol[s]{\target}{\induced{\pmdp}{\sched}}$ for any $s \in S$.
The constraints over variables for all parameters are:
\[
	h_s[\val] \cdot \prod_{s'' \neq s}g_{s''}[\val]\; \leq \;\sum_{s' \in S}\pprob(s, \act, s') \cdot h_{s'}[\val] \cdot \prod_{s'' \neq s'}g_{s''}[\val]
\]
for all $s \in S, \act \in \Act$.
\end{constraints}

\begin{proof}[Proof of Theorem~\ref{thm:fp_ea_reach_in_np}]
	We only give the proof for the $\geq$-relation, the other cases are analogous. Observe that
	\begin{align*}
		& \exists\; \val \in \wdval{\pmdp}, \forall \sched \in \DMSched{\pmdp}: \pr{\instantiated{\pmdp}{\val}}{\sched}{\eventually \target} \geq \nicefrac{1}{2} \\
		\iff & \exists\; \val \in \wdval{\pmdp}: \min_{\sched \in \DMSched{\pmdp}} \pr{\instantiated{\pmdp}{\val}}{\sched}{\eventually \target} \geq \nicefrac{1}{2},
	\end{align*}
  which means that it is sufficient and necessary for the answer to the
  problem to be positive that there be a \emph{somewhere optimal
  strategy}
  which, for
  the valuation for which it is minimal, induces a
  reachability probability of at least $\nicefrac{1}{2}$.  Hence, we may guess a
  somewhere minimal strategy and check its minimality using
  Constraints~\ref{const:aspects:somewhereminimalratfunc} with a conjunction
  that the initial state satisfies the threshold\footnote{Technically, one has to find the zero states and
  	make them sinks. Recall that zero states can be computed using graph-based
  	algorithms for pMDPs and MDPs
  	alike~\cite{BK08}.}. This conjunction only has
  parameters $\pars$, and can thus be checked in \Ptime.
\end{proof}

%% file: figures/encoding_pmdpexample.tex
\begin{tikzpicture}[ every node/.style={scale=1, font=\scriptsize}, st/.style={draw, circle, inner sep=2pt, minimum size=13pt}]
	\node[st] (s00) {$s_1$};
	\node[st, below=of s00] (s10) {$s_0$};
	\node[st, left=of s10] (s20) {$s_2$};
	\node[st, right=3cm of s00, accepting] (s01) {$s_3$};
	\node[st, below=of s01] (s11) {$s_4$};
	%\node[st, below=of s11] (s21) {$s_5$};
	\node[st, right=3cm of s01] (s02) {$s_6$};
	\node[st, below=of s02] (s12) {$s_7$};
	%\node[st, below=of s12] (s22) {$s_8$};
	
	\node[circle, inner sep=1pt, fill=black, above=0.5cm of s02] (a1s02) {};
	\node[circle, inner sep=1pt, fill=black, right=1.5cm of s11,yshift=-0.4cm] (a1s11) {};
	\node[circle, inner sep=1pt, fill=black, right=1.5cm of s11,yshift=0.2cm] (a1s12) {};
	\node[circle, inner sep=1pt, fill=black, above=0.4cm of s12] (a2s12) {};
	
	\node[circle, inner sep=1pt, fill=black, above=0.5cm of s01] (a1s01) {};
	
	%\node[circle, inner sep=1pt, fill=black, above=0.5cm of s00] (a1s00) {};
	\node[circle, inner sep=1pt, fill=black, right=1.5cm of s00] (a2s00) {};

	\node[circle, inner sep=1pt, fill=black, above=0.5cm of s10] (a1s10) {};
	\node[circle, inner sep=1pt, fill=black, right=1.5cm of s10] (a2s10) {};
	
	\node[circle, inner sep=1pt, fill=black, above=0.5cm of s20] (a1s20) {};
	
	%\draw[-] (s00) -- (a1s00);
	%\draw[->] (a1s00) edge[bend left] node[right] {$1$} (s00);
	
	\draw[-] (s20) -- (a1s20);
	\draw[->] (a1s20) edge[bend left] node[right] {$1$} (s20);

	\draw[-] (s10) -- (a2s10);
	\draw[->] (a2s10) edge[bend left] node[below,pos=0.15] {$\frac{1}{3}$} (s20);
	\draw[->] (a2s10) edge node[above] {$\frac{1}{3}$} (s11);
	\draw[->] (a2s10) edge node[above] {$\frac{1}{3}$} (s01);
	
	\draw[-] (s10) -- (a1s10);
	\draw[->] (a1s10) edge node[right] {$x$} (s00);
	\draw[->] (a1s10) edge node[above] {$1-x$} (s20);
	
	\draw[-] (s00) -- (a2s00);
	\draw[->] (a2s00) edge node[above] {$\frac{1}{2}$} (s01);
	\draw[->] (a2s00) edge[bend left] node[right] {$\frac{1}{2}$} (s10);

	\draw[-] (s01) -- (a1s01);
	\draw[->] (a1s01) edge[bend left] node[right] {$1$} (s01);

	\draw[-] (s02) -- (a1s02);
	\draw[->] (a1s02) edge[bend left] node[right] {$x$} (s02);
	\draw[->] (a1s02) edge node[above] {$1-x$} (s01);
	
	\draw[-] (s11) -- (a1s11);
	\draw[->] (a1s11) edge node[below] {$1$} (s12);
	
	\draw[-] (s12) -- (a1s12);
	\draw[->] (a2s12) edge node[above] {$1-x$} (s01);
	\draw[->] (a2s12) edge node[right] {$x$} (s02);
	
	\draw[-] (s12) -- (a2s12);
	\draw[->] (a1s12) edge[bend left] node[above] {$\frac{1}{2}$} (s12);
	\draw[->] (a1s12) edge node[above] {$\frac{1}{2}$} (s11);

\end{tikzpicture}

%% file: figures/encoding_pmdpexamplesmall.tex
\begin{tikzpicture}[ every node/.style={scale=1, font=\scriptsize}, st/.style={draw, circle, inner sep=2pt, minimum size=13pt}]
	\node[st, initial, initial text=] (s0) {$s_1$};
	
	\node[st, right=1.5cm of s0] (s1) {$s_2$};
	\node[st, left=1.5cm of s0] (s2) {$s_0$};
	\node[st, accepting, right=1.5cm of s1] (s3) {$s_3$};
	
	\node[circle, inner sep=1pt, fill=black, above=0.5cm of s0] (a1s0) {};
	\node[circle, inner sep=1pt, fill=black, below=0.5cm of s0] (a2s0) {};
	
	\node[circle, inner sep=1pt, fill=black, above=0.2cm of s1] (a2s1) {};
	
	\draw[-] (s0) -- (a1s0);
	\draw[->] (a1s0) edge[bend left] node[above] {$x$} (s1);
	\draw[->] (a1s0) edge[bend right] node[above] {$1-x$} (s2);
	
	\draw[-] (s0) -- (a2s0);
	\draw[->] (a2s0) edge[bend right] node[above] {$1-y$} (s1);
	\draw[->] (a2s0) edge[bend left] node[above] {$y$} (s2);

	\draw[-] (s1) -- (a2s1);
	\draw[->] (a2s1) edge node[above] {$1-y$} (s0);
	\draw[->] (a2s1) edge node[above] {$y$} (s3);	
\end{tikzpicture}

%% file: figures/complexity_gpvswdequiv.tex
\begin{tikzpicture}[initial text=, every node/.style={font=\scriptsize}, st/.style={draw, circle, inner sep=2pt, minimum size=18pt},]
	\node[initial,st] (s1) {$s_{x_1}$};
	\node[st, right=0.9cm of s1] (s2) {$s'_{x_1}$};
	\node[st, draw, right=0.9cm of s2] (s3) {$s_{x_2}$};
	\node[st, right=0.9cm of s3] (s4) {$s'_{x_2}$};
	\node[circle, right=0.9cm of s4] (s5) {$\hdots$};
	\node[st, right=0.9cm of s5] (s6) {$s'_{x_n}$};
	\node[st, right=0.9cm of s6] (s7) {$\init$};
	\node[right=0.5cm of s7] (s8) {$\hdots$};
	
	\draw[->] (s1) edge node[auto] {$1{-}x_1$} (s2);
	\draw[->] (s2) edge node[auto] {$x_1$} (s3);
	\draw[->] (s3) edge node[auto] {$1{-}x_2$} (s4);
	%\draw[->] (s4) edge node {$p_2$} (s2);
	\draw[->] (s6) edge node[auto] {$x_n$} (s7);

	\draw[->] (s1) edge[loop above] node[auto] {$x_1$} (s1);
	\draw[->] (s2) edge[loop above] node[auto] {$1{-}x_1$} (s2);
	\draw[->] (s3) edge[loop above] node[auto] {$x_2$} (s3);
	\draw[->] (s6) edge[loop above] node[auto] {$1{-}x_n$} (s6);
	\draw[->] (s4) edge[loop above] node[auto] {$1{-}x_2$} (s4);
	\draw[->] (s7) -- (s8);
	
\end{tikzpicture}	

%% file: figures/aspects_frompoltopmc.tex
\begin{tikzpicture}[initial text=,every node/.style={scale=1, font=\scriptsize}]
	\node[state, initial, scale=0.5] (s0) {};
	\node[state, right=of s0,yshift=1.1cm, scale=0.5] (b1) {};
	\node[state, right=of s0,yshift=-1.1cm, scale=0.5] (c1) {};
	\node[state, right=of b1, scale=0.5] (b2) {};
	\node[state, right=of b2, scale=0.5,accepting] (b3) {};
	\node[state, right=of c1, scale=0.5,accepting] (c2) {};
	\node[state, right=2cm of s0, scale=0.5] (bot) {};
	
	\draw[->] (s0) edge node[above] {$\nicefrac{1}{2}$} (b1);
	\draw[->] (s0) edge node[above] {$\nicefrac{1}{6}$} (bot);
	
	\draw[->] (s0) edge node[above] {$\nicefrac{1}{3}$} (c1);
	
	\draw[->] (b1) edge node[above] {$x$} (b2);
	\draw[->] (b1) edge node[left] {$1-x$} (bot);
	
	\draw[->] (b2) edge node[above] {$x$} (b3);
	\draw[->] (b2) edge node[right] {$1-x$} (bot);

	\draw[->] (c1) edge node[above] {$y$} (c2);
	\draw[->] (c1) edge node[right] {$1-y$} (bot);
	
	\draw[->] (bot) edge[loop right] node[right] {$1$} (bot);

\end{tikzpicture}

%% file: figures/aspects_chonev_trick.tex
\begin{tikzpicture}[initial text=,every node/.style={scale=1, font=\scriptsize}]
	\node[state, initial, scale=0.5] (s0) {};
	\node[state, right=of s0,yshift=1.2cm, scale=0.5] (a1) {};
	\node[state, right=of s0,yshift=0.4cm, scale=0.5] (b1) {};
	\node[state, right=of s0,yshift=-0.4cm, scale=0.5] (c1) {};
	\node[state, right=of s0,yshift=-1.2cm, scale=0.5] (d1) {};
	\node[state, right=of a1, scale=0.5] (a2) {};
	\node[state, right=of a2, scale=0.5] (a3) {};
	\node[state, right=of a3, scale=0.5, accepting] (a4) {};
	\node[state, right=of b1, scale=0.5] (b2) {};
	\node[state, right=of b2, scale=0.5,accepting] (b3) {};
	\node[state, right=of c1, scale=0.5,accepting] (c2) {};
	\node[state, right=of d1, scale=0.5,accepting] (d2) {};
	
	\draw[->] (s0) edge[bend left] node[pos=0.75,below] {$\nicefrac{2}{8}$} (a1);
	\draw[->] (s0) edge node[above] {$\nicefrac{2}{8}$} (b1);
	\draw[->] (s0) edge node[above] {$\nicefrac{2}{8}$} (c1);
	\draw[->] (s0) edge[bend right] node[pos=0.75,above] {$\nicefrac{1}{8}$} (d1);
	
	\draw[->] (a1) edge node[above] {$1-x$} (a2);
	\draw[->] (a2) edge node[above] {$x$} (a3);
	\draw[->] (a3) edge node[above] {$y$} (a4);
	
	\draw[->] (b1) edge node[above] {$1-x$} (b2);
	\draw[->] (b2) edge node[above] {$y$} (b3);
	
	\draw[->] (c1) edge node[above] {$1-y$} (c2);
	\draw[->] (d1) edge node[above] {$y$} (d2);

\end{tikzpicture}

%% file: figures/aspects_chonev_trick_proof_idea.tex
\tikzset{decoration={snake,amplitude=.4mm,segment length=2mm,
                       post length=0mm,pre length=0mm}}
	\begin{tikzpicture}[every node/.style={scale=1, font=\scriptsize}]
		\node[state,initial,initial text=,scale=0.6] (s0)  {$ $};
		\node[right= of s0] (dots) {$\vdots$};
		\node[state,above=0.24cm of dots,scale=0.6] (s1) {};
		\node[state,below=0.24cm of dots,scale=0.6] (sm) {};
		\node[state,right=1.5cm of s1,scale=0.6,accepting]    (v1) {};
		\node[state,right=1.5cm of sm,scale=0.6,accepting]    (vm) {};
		
		\draw[->] (s0) edge node[above,xshift=-2mm] {$\tilde{a}_1$} (s1);
		\draw[->] (s0) edge node[below,xshift=-2mm] {$\tilde{a}_m$} (sm);
		\draw[->] (s0) edge[bend left=60] node[above,xshift=-2mm] {$\tilde{b}$} (v1);

		\draw[->,decorate] (s1) -- node[above] {$h_1$} (v1);
		\draw[->,decorate] (sm) -- node[above] {$h_m$} (vm);
		
	\end{tikzpicture}	

%% file: figures/aspects_mdps_etrhard.tex
\begin{tikzpicture}
	\node[circle, draw, initial, initial text=] (sinit) {};
	\node[below=1.3cm of sinit] (dots) {$\dots$};
	\node[rectangle, left=0.7cm of dots, draw, inner sep = 6pt] (f1) {\footnotesize$\pmc(f_1)_{\lambda=\nicefrac{1}{2},\mu=0}$};
	\node[rectangle, right=0.7cm of dots, draw, inner sep = 6pt] (fm) {\footnotesize$\pmc(f_m)_{\lambda=\nicefrac{1}{2},\mu=0}$};
	\node[circle, fill, inner sep = 2pt, above=0.5cm of f1] (a1) {};
	\node[circle, fill, inner sep = 2pt, above=0.5cm of fm] (am) {};
	
	\draw[->] (a1) edge  node[right] {\scriptsize$1$} (f1);
	\draw[->] (am) edge  node[right] {\scriptsize$1$} (fm);
	\draw[->] (sinit) edge  node[right] {} (a1);
	\draw[->] (sinit) edge  node[right] {} (am);

\end{tikzpicture}

%% file: 6_conclusion.tex
We have given a concise overview of known and new results regarding the
complexity of parameter synthesis.  In particular, the new results clarify
that the general case of parameter synthesis is \ETR-complete, as, e.g., asking
whether  (Boolean combination of) polynomials have a common root.  These
results motivate the usage of SMT solvers for ETR to practically solve
parameter synthesis problems. In practice, however, such approaches still lack behind
abstraction-refinement based approaches.
 
Some complexity bounds provided in this paper are not tight.  The most
interesting problem seems to be a lower bound for parameter synthesis in
pMDPs with a single parameter and quantitative reachability.  Another
question is the precise complexity class of parameter synthesis in pMCs with
arbitrarily many parameters and strict bounds on the reachability
probability. 
 
Finally, there seems to be a large zoo of practically relevant subclasses of
pMDP synthesis problems whose complexity may still be explored.